\documentclass[journal]{IEEEtran}

\usepackage{url}
\usepackage{amsmath,amssymb,amsfonts}
\usepackage{amsthm}
\usepackage{mathtools}
\usepackage{tabularx,booktabs, multirow}
\usepackage{graphicx}
\usepackage{bm}
\usepackage{algpseudocode}
\usepackage{algorithm}
\usepackage{physics}
\makeatletter
\def\endfigure{\end@float}
\def\endtable{\end@float}
\makeatother
\usepackage[caption=false]{subfig} 
\usepackage{threeparttable}

\renewenvironment{proof}{{\bfseries Proof.}}{\qedsymbol}
\renewcommand\qedsymbol{$\blacksquare$}
\theoremstyle{definition}

\newtheorem{definitionx}{Definition}
\newtheorem{theoremx}{Theorem}
\newtheorem{lemmax}{Lemma}
\newtheorem{corollaryx}{Corollary}

\newtheorem{remarkx}{Remark}
\newtheorem{propertyx}{Property}

\newcommand{\tp}{^\top}

\DeclareDocumentCommand\diag{}{\opbraces{\operatorname{diag}}}

\DeclareDocumentCommand\rank{}{\opbraces{\operatorname{rank}}}
\DeclareDocumentCommand\tr{}{\opbraces{\operatorname{tr}}}
\DeclareDocumentCommand\vec{}{\opbraces{\operatorname{vec}}}

\DeclareDocumentCommand\edm{}{\opbraces{\operatorname{edm}}}
\begin{document}
\title{Multicluster Design and Control \\ of Large-Scale Affine Formations}

\author{Zhonggang Li, \IEEEmembership{Graduate Student Member, IEEE}, Geert Leus, \IEEEmembership{Fellow, IEEE}, \\ and Raj Thilak Rajan, \IEEEmembership{Senior Member, IEEE}
\thanks{This work was partially presented at the 33rd European Signal Processing Conference (EUSIPCO) in 2025 \cite{li2025fast}, and is partially funded by the Sensor AI Lab, under the AI Labs program of Delft University of Technology and by the European Commission Key Digital Technologies Joint Undertaking - Research and Innovation (HORIZON-KDT-JU-2023-2-RIA), under grant agreement No 101139996, the ShapeFuture project. The authors are with the Signal Processing Systems group, the Faculty of EEMCS, Delft University of Technology, 2628CD, Delft, The Netherlands (email: \{z.li-22, g.j.t.leus, r.t.rajan\}@tudelft.nl).}
}

\maketitle

\begin{abstract}

Conventional affine formation control (AFC) empowers a network of agents with flexible but collective motions - a potential which has not yet been exploited for large-scale swarms. One of the key bottlenecks lies in the design of an interaction graph, characterized by the Laplacian-like stress matrix. Efficient and scalable design solutions often yield suboptimal solutions on various performance metrics, e.g., convergence speed and communication cost, to name a few. The current state-of-the-art algorithms for finding optimal solutions are computationally expensive and therefore not scalable. In this work, we propose a more efficient optimal design for any generic configuration, with the potential to further reduce complexity for a large class of nongeneric rotationally symmetric configurations. Furthermore, we introduce a multicluster control framework that offers an additional scalability improvement, enabling not only collective affine motions as in conventional AFC but also partially independent motions naturally desired for large-scale swarms.
The overall design is compatible with a swarm size of several hundred agents with fast formation convergence, as compared to up to only a few dozen agents by existing methods. Experimentally, we benchmark the performance of our algorithm compared with several state-of-the-art solutions and demonstrate the capabilities of our proposed control strategies.
\end{abstract}


\section{Introduction}\label{sec:intro}

\IEEEPARstart{I}{n} recent years, robotic swarms have attracted significant attention for their ability to execute complex, large-scale missions \cite{zhu2024self, sun2023mean, na2023federated}. These swarms are deployed across a wide variety of domains, from environmental exploration \cite{bartolomei2023fast} and collective payload transport \cite{jurt2022collective}, to synchronized satellite constellations \cite{Zhao2023Satellite}, some of which rely on tightly coordinated behaviors among the individual robots. For these specific missions, where maintaining a spatial relationship during collective motions is essential, formation control serves as an essential mechanism \cite{oh2015survey, li2026geometry, liu2018survey}. A distributed formation control system uses relative information such as interagent distances \cite{babazadeh2022directed}, displacements \cite{fang2023distributed}, bearings \cite{cheng2024general}, etc., to achieve and maintain a desired geometric pattern in 2D or 3D space, and possibly engage in continuous formation maneuvers. Such a system is typically characterized by a framework consisting of two main components: (a) a graph, where the edges denote the communication links for information exchange, and (b) a configuration that instantiates the graph nodes with spatial coordinates. The design of such a framework, similar to domains such as sensor networks \cite{zheng2009wireless}, distributed optimization \cite{heusdens2024distributed}, etc., plays a critical role in determining the overall stability, performance, and communication cost of the system.

Recently, affine formation control (AFC) has garnered significant attention for its maneuverability and coordination flexibility \cite{lin2015necessary, zhao2018affine, yang2024joint}. Unlike traditional rigid formation control, AFC permits continuous geometric transformations—such as scaling, shearing, and rotation—making it advantageous for adaptive mission planning in dynamic environments. Fundamentally, AFC relies on a consensus-based framework where control inputs are derived from the relative displacements among neighboring agents. While this distributed architecture naturally promotes scalability, existing literature rarely addresses the practical challenges of deploying AFC in large-scale swarms. The core of these challenges is the network topology: instead of a standard graph Laplacian, AFC employs a stress matrix \cite{connelly2022frameworks}, a generalized Laplacian with real-valued edge weights intrinsically tied to graph rigidity. Designing a stress matrix is the foundational step for AFC implementation. While the legitimacy of stabilization is a hard constraint, several additional critical aspects should also be considered for large-scale swarms, such as the sparsity of the underlying graph determining the communication overhead, the convergence speed of AFC given an arbitrary initialization, and the computational complexity for practical feasibility.

Designing a stabilizing stress matrix for fast AFC convergence while maintaining a sparse communication topology is challenging, as the communication capacity affects the efficiency of information exchange and thus influences the convergence \cite{kim2005maximizing, olfati2007consensus}. Earlier mathematical literature study rigid geometric designs that guarantee the existence of a stress matrix \cite{kelly2014class, connelly2022frameworks}. Subsequent works formulate numerical feasibility problems for computing such matrices~\cite{lin2015necessary, zhao2018affine}. The focus is on finding a feasible stabilizing stress matrix with a prescribed graph topology, rather than optimizing any objective. This prescription can be impractical for large swarms in 2D and 3D. Similar approaches are adopted in \cite{yang2018constructing} and \cite{li2025distributed}, in which a sparsity pattern of the underlying graph is predetermined. Specifically, \cite{yang2018constructing} adopts a sparse nullspace reconstruction approach that sets specific locations in the stress matrix to zero, and calculates the rest to satisfy the constraints. Expanding graphs is also common in constructing rigid frameworks with stress matrices \cite{li2025distributed, yang2018growing}, where an initial small base graph is considered, and the other nodes are added sequentially by strategically connecting to a few nodes in the existing graph. Both methods yield substantial sparsification compared with a complete graph, resulting in low communication overhead, and admit analytical solutions that are computationally efficient. However, they often rely on the genericity assumption of the given node configuration, and the sparsity pattern or the connecting strategy is heuristic. The critical aspect of convergence speed is not considered either, especially for large networks. A systematic search for an optimal solution with respect to both convergence and communication cost is developed in \cite{xiao2022framework}, in which the existence of edges is modeled by binary variables and the problem is formulated as a mixed-integer semidefinite program (MISDP). Although this framework admits optimized solutions, it raises concerns about the computational complexity and optimality due to the nonconvex formulation. As such, this motivates a stress design method that jointly optimizes the critical objectives while maintaining manageable computational complexity for large networks. Additionally, the absence of systematic comparisons in the literature calls for a quantitative evaluation of existing methods.

In this paper, we present an optimal design of stress matrices through an efficient convex optimization framework for any given configuration. We focus on network sparsification for efficient communication while guaranteeing AFC convergence acceleration. Compared with our preliminary results in \cite{li2025fast}, our additional contributions include the following. 
\begin{itemize}
    \item We reveal the optimal structure of the stress matrix for some special nongeneric geometries, i.e., rotationally symmetric configurations, which further facilitates a reduction in the computational complexity.
    \item To ensure the stress design is compatible with larger-scale networks, we propose a control framework that partitions the swarm into smaller clusters and aligns their motions through bridging nodes. 
    \item We present insights on the stability and provide conditions when the clusters exhibit collective behaviors, as well as when inter-cluster flexibility is introduced.
    \item We perform a comprehensive performance benchmark of our proposed stress design solution against several state-of-the-art solutions across generic and nongeneric test cases ranging from small-scale ($\leq10$ nodes) to large-scale ($\geq100$ nodes) setups in 2D and 3D.
\end{itemize}

The rest of the paper is organized as follows. In Section~\ref{sec: preliminaries}, we introduce the preliminaries of graph rigidity and stress-based AFC. In Section~\ref{sec: methodology}, we propose a convex optimization framework for the stress design and discuss the choice of hyperparameters to balance the objectives. This is considered a generic solution compatible with any given configuration. For symmetric configurations, we discuss their optimal structure under our proposed formulation, and propose the complexity-reduced unique stress identifier (USI) method in Section \ref{sec: usi}. In Section~\ref{sec: mc control}, we introduce the multicluster control framework based on a classical control law, discuss its stability, and provide conditions under which the clusters exhibit fully or partially collective behaviors. Finally, we compare our proposed solutions with the state-of-the-art through a comprehensive benchmark and demonstrate the capability of multicluster control in Section \ref{sec: simulations}. Concluding remarks and future work are briefly outlined in Section \ref{sec: conclusions}.

\textit{Notations.} Vectors and matrices are represented by lowercase and uppercase boldface letters, respectively, such as $\bm{a}$ and $\bm{A}$. Sets and graphs are represented using calligraphic letters, e.g., $\mathcal{A}$. Vectors of length $N$ of all ones and zeros are denoted by $\bm{1}_N$ and $\bm{0}_N$, respectively, and their matrix versions are similarly $\bm{1}_{M\times N}$ and $\bm{0}_{M\times N}$. An identity matrix of size $N$ is denoted by $\bm{I}_N$. The Kronecker product is $\otimes$ and the vectorization operator is $\vec(\cdot)$. The $\diag(\cdot)$ operator creates a diagonal matrix from a vector. Moreover, $\tr(\cdot)$ denotes the trace operator and $\lambda_k(\bm{A})$ denotes the $k$-th smallest eigenvalue of a symmetric matrix $\bm{A}$. We use $\qty[\bm{A}]_{ij}$ and $\qty[\bm{a}]_i$ to denote the elements of matrices and vectors, respectively. Units are enclosed by square brackets, e.g., seconds [s].

\section{Fundamentals}\label{sec: preliminaries}

\subsection{Graphs and Rigidity Theory}
Consider $N$ mobile agents in $D$-dimensional Euclidean space where $N\geq D+1$. They are interconnected with a communication network, modeled by an undirected graph $\mathcal{G} = (\mathcal{V}, \mathcal{E})$, where the vertices $\mathcal{V} = \qty{1,...,N}$ denote the agents, and the edges $\mathcal{E}\subseteq\mathcal{V}\times \mathcal{V}$ denote the communication links. We use $N = \qty|\mathcal{V}|$ and $M = \qty|\mathcal{E}|$ as shorthand for the number of vertices and edges, respectively. The set of neighbors of a node $i$ is defined as $\mathcal{N}_i = \{j\in\mathcal{V}: (i,j)\in\mathcal{E}\}$. Each node $i\in\mathcal{V}$ is assigned a position $\bm{p}_i\in\mathbb{R}^D$, and the \textit{configuration} of the nodes in $\mathcal{V}$ is $\bm{P}(\mathcal{V}) = [\bm{p}_1, \dots, \bm{p}_N] \in \mathbb{R}^{D \times N}$. We also define an augmented configuration $\bar{\bm{P}}(\mathcal{V}) = \qty[\bm{P}\tp, \bm{1}_N]\tp\in\mathbb{R}^{(D+1)\times N}$. They will be shorthanded to $\bm{P}$ and $\bar{\bm{P}}$ respectively, unless intended for particular sets of vertices. A \textit{generic configuration} \cite{connelly2005generic} has algebraically independent node coordinates, i.e., there are no geometric constraints among nodes. A configuration of random node coordinates has probability one of being generic. Typical nongeneric configurations include grids, a regular polygon, etc. 

A \textit{framework} $\mathcal{F} = (\mathcal{G}, \bm{P})$ pairs the graph $\mathcal{G}$ with its configuration $\bm{P}$. Two frameworks $(\mathcal{G}, \bm{P})$ and $(\mathcal{G}, \bm{P}')$ are \textit{equivalent} if $\norm{\bm{p}_i - \bm{p}_j}_2 = \norm{\bm{p}'_i - \bm{p}'_j}_2$ for all $(i,j) \in \mathcal{E}$, and \textit{congruent} if this distance equality holds for all node pairs $(i,j) \in \mathcal{V} \times \mathcal{V}$.  A framework $\mathcal{F}$ is \textit{globally rigid} (often termed \textit{rigid}) if every equivalent framework is also congruent. Furthermore, $\mathcal{F}$ is \textit{universally rigid} if it remains globally rigid in any higher-dimensional embedding $\mathbb{R}^{D'}$ where $D' \geq D$. Fig. \ref{fig: urf examples} illustrates these rigidity concepts. Universal rigidity can be algebraically represented by an \textit{equilibrium stress} $\omega_{ij}\in\mathbb{R}$ for every edge $(i,j) \in \mathcal{E}$, which together satisfy
\begin{equation}\label{equ: stress def}
    \sum_{(i,j)\in\mathcal{E}}\omega_{ij}(\bm{p}_i-\bm{p}_j) = \bm{0}_D.
\end{equation}
A more compact form of (\ref{equ: stress def}) is $\bm{\Omega}\bm{P}\tp=\bm{0}_{N\times D}$, where $\bm{\Omega}\in\mathbb{R}^{N\times N}$ is the \textit{stress matrix} defined as
\begin{equation}\label{equ: def stress matrix}
    [\bm{\Omega}]_{ij}=\begin{cases}
  0,& \text{if } (i,j)\notin\mathcal{E}\\
  -\omega_{ij},& \text{if } i\neq j, (i,j)\in\mathcal{E}\\
  \sum_{k\in\mathcal{N}_i}\omega_{ik}&  \text{if } i=j
\end{cases}.
\end{equation}
Alternatively, the equilibrium stress and the stress matrix are linked in matrix-vector form through the graph incidence matrix $\bm{B}\in\mathbb{R}^{N\times M}$
\begin{equation}\label{equ: stress inc}
    \bm{\Omega} = \bm{B}\diag(\bm{\omega})\bm{B}\tp,
\end{equation}
where $\bm{\omega}\in\mathbb{R}^M$ is a vector containing all the equilibrium stresses. Note that $\bm{\Omega}$ reduces to a standard graph Laplacian if $\diag(\bm{\omega})=\bm{I}_M$, i.e., equal weights for the edges, and it also has $\bm{1}_N$ in the nullspace like the Laplacian. The following theorem establishes the important properties of $\bm{\Omega}$ related to universal rigidity.
\begin{theoremx}\label{thm: rigidity and stress}
(\textit{Universally Rigid Frameworks and Stress Matrices}) Given a framework $\mathcal{F} = (\mathcal{G},\bm{P})$ with $\bm{P}$ being a generic configuration, $\mathcal{F}$ is universally rigid if and only if there exists a stabilizable stress matrix $\bm{\Omega}$, i.e., a positive semidefinite (PSD) stress matrix with rank $N-D-1$.
\end{theoremx}
\noindent\begin{proof}
    See \cite{gortler2014characterizing}, \cite{zhao2018affine}, \cite{xiao2022framework}.
\end{proof}
\noindent 

\noindent Note that a generic configuration is typically assumed by the literature for mathematical guarantees, but nongeneric geometries are also useful for formation control applications. It is worth mentioning that our proposed stress design approaches work for both generic and nongeneric configurations.

\begin{figure}[t]
	\centering	

        \subfloat[\scriptsize flexible]{\raisebox{0ex}
		{\includegraphics[width=0.13\textwidth]{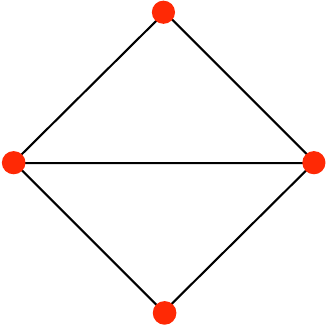}}%
	}
        \hspace{3mm}
	\subfloat[\scriptsize globally rigid]{\raisebox{0ex}
		{\includegraphics[width=0.13\textwidth]{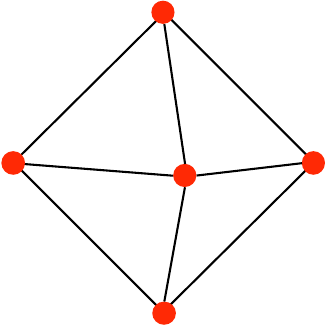}}%
	}
    \hspace{3mm}
        \subfloat[\scriptsize universally rigid]{\raisebox{0ex}
    		{\includegraphics[width=0.13\textwidth]{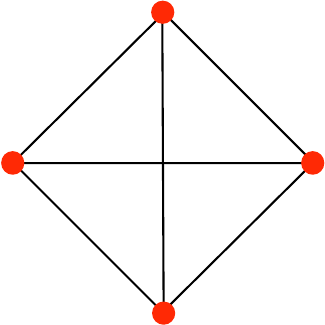}}%
    	}

	\caption{\small Examples of frameworks in $\mathbb{R}^2$ with increasing rigidity \cite{lin2015necessary}}
	\label{fig: urf examples}
		
\end{figure}

\subsection{Single-Cluster Affine Formation Control}
Given a stress matrix $\bm{\Omega}$ under a universally rigid framework, affine formation control (AFC) tries to steer the real-time position $\bm{z}_i\in\mathbb{R}^D$ into a target position $\bm{p}_i$ for agent $i\in\mathcal{V}$. To achieve this, we assume agents adopt single-integrator dynamics $\dot{\bm{z}}_i = \bm{u}_i$ where $\bm{u}_i$ is a velocity control input. The conventional consensus-based control law is designed as \cite{lin2015necessary}
\begin{equation}\label{equ: conventional local law}
    \bm{u}_i = -\sum_{j\in\mathcal{N}_i}[\bm{\Omega}]
    _{ij}(\bm{z}_j-\bm{z}_i),
\end{equation}
which has been extended and enriched for several scenarios \cite{zhao2018affine,xu2018affine}. The collective dynamics of (\ref{equ: conventional local law}) can be described by the linear time-invariant (LTI) system
\begin{equation}\label{equ: global dynamics}
    \dot{\bm{z}} = 
-\qty(\bm{\Omega}\otimes\bm{I}_D)\bm{z},
\end{equation}
where $\bm{z} = \qty[\bm{z}_1\tp,...,\bm{z}_N\tp]\tp\in\mathbb{R}^{DN}$. Recall that the stress matrix from (\ref{equ: stress def}) satisfies $\bm{\Omega}\bm{P}\tp=\bm{0}$, which can be rewritten as $(\bm{\Omega}\otimes\bm{I}_D)\bm{p}=\bm{0}$, where $\bm{p} = \vec(\bm{P})$. It is then straightforward to see that $\bm{p}$ is an equilibrium of system (\ref{equ: global dynamics}). This is considered an extension of the average consensus system using the graph Laplacian, which has only a $1$-dimensional nullspace corresponding to the $\bm{1}_N$ vector, i.e., only configurations with all nodes in the same location can be an equilibrium.

Moreover, since $\bm{\Omega}$ is positive semidefinite, system (\ref{equ: global dynamics}) is globally and exponentially stable \cite{zhao2018affine}, meaning that the formation asymptotically converges to the equilibrium space containing $\bm{p}$ given any initialization. Note that the targeted equilibrium $\bm{p}$ can be reached up to an affine transformation (hence its usefulness for AFC), and can be uniquely determined given a minimum of $D+1$ anchor nodes, which are commonly named leaders. Conceptually, they tie up the loose degrees of freedom from the flexible affine transformations and thus can guide continuous formation maneuvers. This is favored since only a small set of leaders needs to be manipulated, while the majority remains fully autonomous. However, the collective affine flexibility may still be considered limited in some cases for large-scale formation as a whole.

\subsection{Problem Formulation}
Given a potentially large set of agents $\mathcal{V}$ ($N\geq100$) with any configuration $\bm{P}$, we primarily seek a stabilizable stress matrix that yields fast convergence for system (\ref{equ: global dynamics}) and a sparse structure for lower communication load. In consensus theory, the second smallest eigenvalue of the Laplacian (Fiedler value) governs convergence speed \cite{olfati2007consensus}, as it determines the slowest decaying component of the error. Prior works have aimed to maximize this eigenvalue for faster convergence \cite{kim2005maximizing}. Similarly, we aim to maximize the smallest nonzero eigenvalue, i.e., the ($D+2$)-th smallest eigenvalue of ${\bm{\Omega}}$ or $\lambda_{D+2}({\bm{\Omega}})$ for stabilizable stress matrices. Additionally, we seek a control framework that preserves the affine flexibility but could introduce additional degrees of freedom beyond collective affine transformations. 

\section{Optimal Topology Design}\label{sec: methodology}
In this section, we focus on optimal stress matrix design using a convex formulation, considering all the nodes as one group. 
We first present a preliminary formulation that directly reflects the objectives and constraints but is intractable. We then reformulate it into a tractable convex optimization problem and provide insights into the introduced hyperparameters.
Recall from (\ref{equ: stress inc}) that the stress matrix can be constructed using the graph incidence matrix given a stress vector. We initialize the graph as a complete graph with incidence matrix $\bar{\bm{B}}\in\mathbb{R}^{N\times\bar{M}}$ having $\bar{M} = \frac{N(N-1)}{2}$ edges. We then aim to find a sparse vector $\bar{\bm{\omega}}\in\mathbb{R}^{\bar{M}}$ such that $\bm{\Omega}$ is sparse by (\ref{equ: stress inc}). The effective number of edges will be $M = \norm{\bar{\bm{\omega}}}_0$. Hence we pose the following problem $\mathcal{P}_0$:

\begin{subequations} \label{equ: primitive formulation}
  \begin{align}
  & \mathcal{P}_0: \quad \underset{\bar{\bm{\omega}},\bm{\Omega}}{\text{minimize}}
  && \norm{\bar{\bm{\omega}}}_0 - \alpha\lambda_{D+2}(\bm{\Omega}) \label{equ: obj pri} \\
  & \quad \quad \quad \text{subject to}
  && \bm{\Omega}=\bar{\bm{B}}\diag(\bar{\bm{\omega}})\bar{\bm{B}}\tp \label{equ: equality cstr pri}\\
  &&& \rank(\bm{\Omega}) = N-D-1 \label{equ: rank cstr pri}\\
  &&& \bm{\Omega}\succeq 0 \label{equ: PSD cstr pri}\\
  &&& \bm{\Omega}\bar{\bm{P}}\tp=\bm{0} \label{equ: null cstr pri}
  \end{align}
\end{subequations} where $\alpha$ is a weighting parameter between the network sparsity by $\norm{\bar{\bm{\omega}}}_0$ and the convergence speed by $\lambda_{D+2}(\bm{\Omega})$. As can be observed, $\mathcal{P}_0$ is difficult to solve due to the L0-norm and the rank constraint. In the next section, we convexify this problem to a tractable semidefinite program (SDP).

\subsection{Convexification and Eigenvalue Maximization}
A common convex relaxation for promoting sparsity replaces the L0-norm with an L1-norm, i.e., we relax $\norm{\bar{\bm{\omega}}}_0$ to $\norm{\bar{\bm{\omega}}}_1$ in the objective. Further, assuming a full-rank $\bar{\bm{P}}$, we let  $\bm{Q}\in\mathbb{R}^{N\times (N-D-1)}$ with orthonormal columns span the kernel of $\bar{\bm{P}}$. Then, (\ref{equ: rank cstr pri}), (\ref{equ: PSD cstr pri}), and (\ref{equ: null cstr pri}) imply that $\bm{Q}\tp\bm{\Omega}\bm{Q}\in\mathbb{R}^{(N-D-1)\times(N-D-1)}$ is a positive definite (PD) matrix, since the null-subspace of $\bm{\Omega}$ is projected out by $\bm{Q}$, i.e., $\bm{Q}\tp\bm{\Omega}\bm{Q}\succ 0$. As such, this PD constraint is equivalent to (\ref{equ: rank cstr pri}) under (\ref{equ: PSD cstr pri}) and (\ref{equ: null cstr pri}).

To maximize $\lambda_{D+2}(\bm{\Omega})$, we explicitly use the trace of the stress matrix $\tr(\bm{\Omega})=\tr\qty(\bm{Q}\tp\bm{\Omega}\bm{Q})$, a convex proxy which effectively sums all the (nonzero) eigenvalues of $\bm{\Omega}$, together with a spectral norm constraint $\norm{\bm{\Omega}}_2\leq\beta$, where we limit the largest eigenvalue of $\bm{\Omega}$ to $\beta>0$. Without additional objectives and constraints, $\lambda_{D+2}(\bm{\Omega})$ is effectively maximized to $\lambda_{D+2}(\bm{\Omega})=...=\lambda_{\text{max}}(\bm{\Omega}) = \beta$ with a minimum condition number $\kappa(\bm{\Omega})=1$. In our formulation, this boosting effect is further shaped by the sparsity term and other constraints. It is worth mentioning that trace regularization is also commonly seen in network design literature, such as \cite{kim2005maximizing}. Additionally, a low condition number has been shown to improve the robustness of AFC against time delays \cite{xiao2022framework} and hence is aligned with our objective for convergence acceleration.

Having the above remodeling, a working convex formulation is shown in $\mathcal{P}_1$
\begin{subequations}\label{equ: cvx form mat var}
  \begin{align}
  & \mathcal{P}_1: \quad \underset{\bar{\bm{\omega}},\bm{\Omega}}{\text{minimize}}
  && f(\bar{\bm{\omega}}) = \norm{\bar{\bm{\omega}}}_1 - \alpha\tr(\bm{\Omega}) \label{equ: obj func}\\
  & \quad \quad \quad \text{subject to}
  && \bm{\Omega}=\bar{\bm{B}}\diag(\bar{\bm{\omega}})\bar{\bm{B}}\tp \label{equ: stress matvec} \\
  &&& \norm{\bm{\Omega}}_2 \leq \beta \label{equ: cstr max eig}  \\
  &&& \bm{Q}\tp\bm{\Omega}\bm{Q}-\gamma\bm{I}_{N-D-1}\succeq 0 \label{equ: cstr min eig}\\
  &&& \bm{\Omega}\bar{\bm{P}}\tp=\bm{0} \label{equ: null cstr mat}
  \end{align}
\end{subequations}
where $\alpha>0$ is the weighting parameter of the objectives and $\beta>0$ upper bounds the spectral norm of $\bm{\Omega}$. We also set a lower bound $\gamma>0$ for (\ref{equ: cstr min eig}) for two explicit reasons, (a) positive-definiteness $\bm{Q}\tp\bm{\Omega}\bm{Q}\succ 0$ cannot be strictly enforced in numerical solvers; and (b) when the L1 term is dominating, we want to avoid the trivial solution $\bar{\bm{\omega}}=\bm{0}$. Note that $\mathcal{P}_1$ can be solved by off-the-shelf solvers, but could be cumbersome because the problem scales badly, as both $\bar{\bm{\omega}}$ and $\bm{\Omega}$ are optimization variables while linked through an equality constraint. Next, we further simplify $\mathcal{P}_1$ by using only $\bar{\bm{\omega}}$ as the optimization variable, yielding an equivalent but efficient final formulation.

\subsection{Proposed Framework}
Substituting (\ref{equ: stress matvec}) into  (\ref{equ: null cstr mat}), we have $\bar{\bm{P}}\bar{\bm{B}}\diag\qty(\bar{\bm{\omega}})\bar{\bm{B}}\tp = \bar{\bm{P}}\bar{\bm{B}}\diag\qty(\bar{\bm{\omega}})\qty[\bar{\bm{b}}_1,...,\bar{\bm{b}}_N] = \bm{0}$ where $\bar{\bm{b}}_i\in\mathbb{R}^{\bar{M}}, \forall i\in\mathcal{V}$ is the $i$-th column of $\bar{\bm{B}}\tp$. Observe that $\bar{\bm{P}}\bar{\bm{B}}\diag\qty(\bar{\bm{\omega}})\bar{\bm{b}}_i = \bar{\bm{P}}\bar{\bm{B}}\diag\qty(\bar{\bm{b}}_i)\bar{\bm{\omega}}=\bm{0}$. As such, we can construct a matrix $\bm{E}\in\mathbb{R}^{N(D+1)\times \bar{M}}$ with structure
\begin{equation}\label{equ: Zhao2018E}
    \bm{E} = \mqty[\bar{\bm{P}}\bar{\bm{B}}\diag\qty(\bar{\bm{b}}_1) \\ \vdots \\ \bar{\bm{P}}\bar{\bm{B}}\diag\qty(\bar{\bm{b}}_N)],
\end{equation}
such that $\bm{E}\bar{\bm{\omega}} = \bm{0}$, which can replace constraint (\ref{equ: null cstr mat}). Next, we combine (\ref{equ: stress matvec}) and (\ref{equ: cstr min eig}) such that  $\bm{Q}\tp\bm{\Omega}\bm{Q} = \bm{\Psi}\diag(\bar{\bm{\omega}})\bm{\Psi}\tp\succ 0$, where $\bm{\Psi} = \bm{Q}\tp\bar{\bm{B}}$. We can then also write
\begin{align}\label{equ: trace linear}
    \tr\qty(\bm{Q}\tp\bm{\Omega}\bm{Q}) 
  &=  \tr\qty(\bm{\Psi}\diag(\bar{\bm{\omega}})\bm{\Psi}\tp) = \tr\qty(\bm{\Psi}\tp\bm{\Psi}\diag(\bar{\bm{\omega}})) \\ \notag
  &= \tr\qty(\diag(\bm{\Psi}\tp\bm{\Psi})\bar{\bm{\omega}}) = \bm{\psi}\tp\bar{\bm{\omega}},
\end{align}
where $\bm{\psi} = \diag(\bm{\Psi}\tp\bm{\Psi})$. As such, we could rewrite $\mathcal{P}_1$ into an equivalent form $\mathcal{P}_2$
\begin{subequations} \label{equ: main formulation}
  \begin{align}
  & \mathcal{P}_2: \quad \underset{\bar{\bm{\omega}}}{\text{minimize}}
  && \norm{\bar{\bm{\omega}}}_1 - \alpha\bm{\psi}\tp\bar{\bm{\omega}} \label{equ: obj main} \\
  & \quad \quad \quad \text{subject to}
  &&\bm{\Psi}\diag(\bar{\bm{\omega}})\bm{\Psi}\tp - \gamma\bm{I}_{N-D-1} \succeq 0 \label{equ: PSD main}\\
  &&& \norm{\bar{\bm{B}}\diag(\bar{\bm{\omega}})\bar{\bm{B}}\tp}_2 \leq \beta \label{equ: 2norm main}\\
  &&& \bm{E}\bar{\bm{\omega}} = \bm{0} \label{equ: null main}
  \end{align}
\end{subequations}
where there is only $\bar{\bm{\omega}}$ as optimization variable. Next, we discuss the choice of hyperparameters, namely $\alpha$, $\beta$, and $\gamma$ in the formulation.

\subsection{Choice of Hyperparameters}\label{sec: hparams}
The need for $\beta > \gamma$ is straightforward to ensure feasibility since they represent the largest and smallest eigenvalues of $\bm{\Omega}$, respectively. They should also be sufficiently apart for a relatively large feasible region. The objective function contains a sparsity term $\norm{\bar{\bm{\omega}}}_1$ (geometrically a cone) and a convergence term $-\alpha\bm{\psi}\tp\bar{\bm{\omega}}$ (a hyperplane). If the convergence term is dominant, i.e., $\alpha$ is large, the hyperplane will unfold the cone such that its global minimum vanishes, in which case the solution is minimized at the boundary of the feasible region, which is not necessarily a sparse solution. Mathematically, to ensure the existence of a global minimum of the objective function, $\bm{0}_{\bar{M}}$ should be in the subdifferential of the objective function, i.e.,
\begin{equation}
    \bm{0}_{\bar{M}} \in \partial\qty(\norm{\bar{\bm{\omega}}}_1 - \alpha\bm{\psi}\tp\bar{\bm{\omega}}),
\end{equation}
which directly yields that all elements of $\alpha\bm{\psi}$ are in the range $(-1,1)$. Recalling that all elements of $\bm{\psi}$ are nonnegative from the definition $\bm{\psi} = \diag(\bm{\Psi}\tp\bm{\Psi})$ and that $\alpha >0$, we could thus define a critical value $\alpha^* = {1}/{\norm{\bm{\psi}}_\infty}$, where the infinity norm is the largest value of $\bm{\psi}$. The solution will be sparse with maximally accelerated convergence at the critical value $\alpha^*$, and the sparsity will trade off with faster convergence for larger $\alpha$. Note that values below $\alpha^*$ do not necessarily yield a sparser solution. Rather, the dominant L1-norm forces all edge weights closer to zero, including the active stress values, which makes it difficult to extract the graph topology by setting a threshold.

\section{The Unique Stress Identifier (USI) Method}\label{sec: usi}
The previous section proposes a convex formulation to design the stress matrix for any given configuration. In this section, we narrow our scope down to a class of highly symmetric configurations in 2D and 3D, namely, rotationally symmetric configurations, with several examples shown in Fig. \ref{fig: rot-symm configs}.  We reveal that the optimal stress matrices for these configurations are highly structured, and the stress values are also highly repetitive. Therefore, the number of distinct values in the stress for optimization significantly decreases, which motivates the design of a shortcut formulation that optimizes only the unique values of the stress with lower computational complexity. We name this approach the unique stress identifier (USI) method.

\subsection{Optimal Stress for Symmetric Configurations}
Rotationally symmetric configurations are tightly coupled with node permutations, which is a crucial property we leverage to reveal the optimal stress structure. These configurations are formalized in the following definition.

\begin{definitionx}[\textit{Rotationally symmetric configurations}]\label{def: rot-symm configs}
    A rotationally symmetric configuration $\bm{P}\in\mathbb{R}^{D\times N}$ satisfies 
    \begin{equation}\label{equ: rot perm equivalence}
        \bm{R}_k\bm{P} = \bm{P}\bm{\Pi}_k,
    \end{equation}   
    where $\bm{R}_k\in\mathbb{SO}(D)$ and $\bm{\Pi}_k\in\qty{0,1}^{N\times N}$, for $k=1,...,K$, is the $k$-th rotation-permutation pair that preserves the configuration. This means that rotationally symmetric configurations are invariant to rotations up to a reordering of the vertices. 
\end{definitionx}

\begin{figure}[t]
    \centering
    \subfloat[\scriptsize An octagon in 2D]{\raisebox{0mm}{ 
        \includegraphics[width=0.14\textwidth]{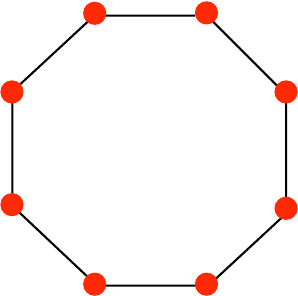}
    }}
    \hspace{0pt} 
    \subfloat[\scriptsize An octahedron in 3D.]{
        \includegraphics[width=0.15\textwidth]{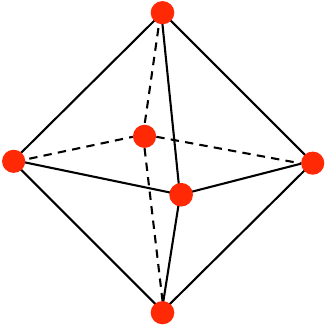}
    }
    \hspace{0pt} 
    \subfloat[\scriptsize A cuboctahedron in 3D.]{\raisebox{0.5mm}{
        \includegraphics[width=0.14\textwidth]{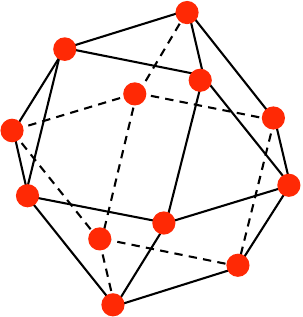}
    }}
    
    \caption{A few examples of rotationally symmetric configurations. Note that the optimal stress for these configurations also applies to their affine transformations.}
    \label{fig: rot-symm configs}
\end{figure}

Since the stress matrix $\bm{\Omega}$ coexists with the stress vector $\bm{\omega}$ through the graph incidence matrix (\ref{equ: stress inc}), there exists a corresponding edge permutation $\bar{\bm{\Pi}}_{\mathcal{E},k}\in\qty{0,1}^{\bar{M}\times \bar{M}}$ for the complete graph characterized by the incidence matrix $\bar{\bm{B}}$ to a given node permutation $\bm{\Pi}_k$, such that 
\begin{equation}\label{equ: node edge perm}
    \bm{\Pi}_k\bm{\Omega}{\bm{\Pi}_k}\tp = \bar{\bm{B}}\diag\qty(\bar{\bm{\Pi}}_{\mathcal{E},k}\bar{\bm{\omega}})\bar{\bm{B}}\tp.
\end{equation}
\noindent This is considered the equivalence of permutations in the node and edge domain. Next, we build up towards the structure of optimal stress matrices under symmetric configurations. For mathematical simplicity, we focus on $\mathcal{P}_1$ (\ref{equ: cvx form mat var}) instead of the equivalent $\mathcal{P}_2$ (\ref{equ: main formulation}).

\begin{definitionx}[\textit{Permutation invariance of feasible sets of $\mathcal{P}_1$}]
Denote the feasible set of $\mathcal{P}_1$ as $\mathcal{C}$. Given a rotationally symmetric $\bm{P}$ as in Definition \ref{def: rot-symm configs}, then the feasible set $\mathcal{C}$ is permutation invariant if, for any feasible $\bar{\bm{\omega}}\in\mathcal{C}$, all the $K$ permutations are also feasible, i.e., $\bar{\bm{\Pi}}_{\mathcal{E},k}\bar{\bm{w}}\in\mathcal{C},\ k=1,...,K$, where $\bar{\bm{\Pi}}_{\mathcal{E},k}\in\qty{0,1}^{\bar{M}\times\bar{M}}$ is an edge permutation related to the $k$-th symmetry rotation.
\end{definitionx}

\begin{definitionx}[\textit{Permutation invariance of objectives of $\mathcal{P}_1$}]\label{def: obj invariance}
Define the objective function of $\mathcal{P}_1$ as $f(\bar{\bm{\omega}})$, where $f: \mathbb{R}^{\bar{M}}\mapsto\mathbb{R}$. Given a rotationally $\bm{P}$ in Definition \ref{def: rot-symm configs}, $f(\bar{\bm{\omega}})$ is permutation invariant if the permutation of the argument $\bar{\bm{\omega}}$ results in the same objective value, i.e., $f(\bar{\bm{\Pi}}_{\mathcal{E},k}\bar{\bm{\omega}}) = f(\bar{\bm{\omega}}), k=1,...,K$.
\end{definitionx}

\begin{lemmax}[\textit{Permutation invariance of $\mathcal{P}_1$}]\label{lmm: prob perm invariance}
    The objective function (\ref{equ: obj func}) and the feasible set $\mathcal{C}$ defined by constraints (\ref{equ: stress matvec}) to (\ref{equ: null cstr mat}) of $\mathcal{P}_1$ are both permutation invariant.
\end{lemmax}
\noindent\begin{proof}
    See Appendix \ref{sec: properties}.
\end{proof}

Lemma \ref{lmm: prob perm invariance} establishes that if a stress matrix for rotationally symmetric configurations is feasible under $\mathcal{P}_1$, then applying permutations to the stress will preserve the feasibility and the cost value. We then define the permutation-invariant stress matrices, which we show are optimal under $\mathcal{P}_1$.

\begin{definitionx}[\textit{Permutation-invariant stress matrices $\bm{\Omega}_\mathrm{pi}$}]
For a given rotationally symmetric configuration, a permutation-invariant $\bm{\Omega}_\mathrm{pi}\in\mathbb{R}^{N\times N}$ with its vector counterpart $\bar{\bm{\omega}}_{\mathrm{pi}}\in\mathbb{R}^{\bar{M}}$ satisfying $\bm{\Omega}_{\mathrm{pi}} = \bar{\bm{B}}\diag(\bar{\bm{\omega}}_{\mathrm{pi}})\bar{\bm{B}}\tp$, satisfies 
\begin{equation}\label{equ: symmetric sol}
    \bm{\Omega}_\mathrm{pi} = \bm{\Pi}_k\bm{\Omega}_\mathrm{pi}{\bm{\Pi}_k}\tp,
\end{equation}
where $\bm{\Pi}_k$ for $k=1,...,K$ denotes the permutation from (\ref{equ: rot perm equivalence}) in Definition \ref{def: rot-symm configs}.
\end{definitionx}


\begin{theoremx}[\textit{Optimality of permutation-invariant stress}]\label{thm: opt symm sol}
    Given a rotationally symmetric configuration, there exists a permutation-invariant $\bm{\Omega}_{\mathrm{pi}}$ that is feasible and optimal for $\mathcal{P}_1$.
\end{theoremx}

\noindent\begin{proof}
Assume $\bar{\bm{w}}^*\in\mathbb{R}^{\bar{M}}$ is a minimizer of $\mathcal{P}_1$ with its matrix counterpart $\bm{\Omega}^* = \bar{\bm{B}}\diag(\bar{\bm{\omega}}^*)\bar{\bm{B}}\tp$, which is not necessarily permutation-invariant. Using the invariance of feasible sets and objective functions in Lemma \ref{lmm: prob perm invariance}, we can find $K$ feasible stress by permuting the configuration with the same objective value, i.e., $\bar{\bm{\Pi}}_{\mathcal{E},k}\bar{\bm{\omega}}^*$ or $\bm{\Pi}_{k}\bm{\Omega}^*\bm{\Pi}_{k}\tp$ for $k=1,...,K$. Then we construct a permutation-invariant stress by averaging across all the permutations, i.e., $\bm{\Omega}_\mathrm{pi} = \frac{1}{K}\sum_{k=1}^{K}\bm{\Pi}_k\bm{\Omega}^*{\bm{\Pi}_k}\tp$ and
\begin{equation}
    \bar{\bm{\omega}}_{\mathrm{pi}} = \frac{1}{K}\sum_{k=1}^{K}\bar{\bm{\Pi}}_{\mathcal{E},k}\bar{\bm{\omega}}^*,
\end{equation}
which is also feasible since it is a convex combination of feasible points. Using Jensen's inequality for convex functions, observe that
\begin{equation}
    f\qty(\bar{\bm{\omega}}_{\mathrm{pi}}) = f\qty(\frac{1}{K}\sum_{k=1}^{K}\bar{\bm{\Pi}}_{\mathcal{E},k}\bar{\bm{\omega}}^*)\leq\frac{1}{K}\sum_{k=1}^{K}f\qty(\bar{\bm{\Pi}}_{\mathcal{E},k}\bar{\bm{\omega}}^*).
    \end{equation}
Since the permutations have the same objective from the invariance of $\mathcal{P}_1$, we know $f\qty(\bar{\bm{\Pi}}_{\mathcal{E},k}\bar{\bm{\omega}}^*) = f(\bar{\bm{\omega}}^*)$ for all $k$, which can be substituted into the inequality and yield
\begin{equation}
    f\qty(\bar{\bm{\omega}}_{\mathrm{pi}}) \leq \frac{1}{K}\sum_{k=1}^{K}f(\bar{\bm{\omega}}^*) = \frac{1}{K}Kf(\bar{\bm{\omega}}^*) = f(\bar{\bm{\omega}}^*).
\end{equation}
This means a permutation-invariant stress will achieve equal or lower cost than the assumed minimizer with arbitrary structure.  In other words, the permutation-invariant stress is the optimal solution to $\mathcal{P}_1$.
\end{proof}

\noindent From Theorem \ref{thm: opt symm sol} it is guaranteed that problem $\mathcal{P}_1$ (or equivalently $\mathcal{P}_2$) under a rotationally symmetric configuration $\bm{P}$ following (\ref{equ: rot perm equivalence}), has an optimal permutation-invariant solution $\bm{\Omega}_{\mathrm{pi}}$. Next, we show that $\bm{\Omega}_{\mathrm{pi}}$ has identical stress values if their corresponding edges in the complete graph can be mapped from one to the other using a symmetric rotation-permutation pair, just like the invariance (\ref{equ: symmetric sol}) suggests.
\begin{corollaryx}[\textit{Edge equivalence in the stress matrix}]\label{crl: edge equivalence}
    Let $\bm{P}$ be a rotationally symmetric configuration and $\bm{\Omega}_{\mathrm{pi}}$ satisfying (\ref{equ: symmetric sol}) be the optimal solution. Two vertex pairs, $(i,j)$ and $(g,l)$, are defined as \emph{symmetrically equivalent} if there exists a symmetry rotation-permutation pair $(\bm{R}_k,\bm{\Pi}_k)$ that maps vertex $i$ to $g$ and vertex $j$ to $l$.

    If the vertex pair $(i,j)$ is symmetrically equivalent to the vertex pair $(g,l)$, their corresponding stresses are equal, i.e.,
    \begin{equation}\label{equ: stress symm equivalence}
        \qty[\bm{\Omega}_\mathrm{pi}]_{ij}=\qty[\bm{\Omega}_\mathrm{pi}]_{gl}.
    \end{equation}

\end{corollaryx}
\noindent\begin{proof}
    If permutation $\bm{\Pi}_k$ maps index $i$ to $g$ and $j$ to $l$, then $\qty[\bm{\Pi}_k]_{gi}=1$ and $\qty[\bm{\Pi}_k]_{lj}=1$. Since $\bm{\Pi}_k$ is a symmetric permutation, it holds that $\bm{\Omega}_\mathrm{pi} = \bm{\Pi}_k\bm{\Omega}_\mathrm{pi}{\bm{\Pi}_k}\tp$, which can be written in its matrix element form as
    \begin{equation}\label{equ: stress symm mat element}
        \qty[\bm{\Pi}_k\bm{\Omega}_\mathrm{pi}\bm{\Pi}_k\tp]_{gl} = \sum_{p=1}^N\sum_{q=1}^N\qty[\bm{\Pi}_k]_{gp}\qty[\bm{\Omega}_\mathrm{pi}]_{pq}\qty[\bm{\Pi}_k\tp]_{ql}.
    \end{equation}
    Since $\qty[\bm{\Pi}_k]_{gp}=1$ only when $p=i$ and otherwise is zero, and $\qty[\bm{\Pi}_k\tp]_{ql} = \qty[\bm{\Pi}_k]_{lq}=1$ only when $q=j$ and otherwise is zero, (\ref{equ: stress symm mat element}) simplifies to $\qty[\bm{\Pi}_k\bm{\Omega}_\mathrm{pi}\bm{\Pi}_k\tp]_{gl}=\qty[\bm{\Omega}_\mathrm{pi}]_{ij}=\qty[\bm{\Omega}_\mathrm{pi}]_{gl}$. Hence (\ref{equ: stress symm equivalence}) is proved.
\end{proof}
    
\noindent Corollary \ref{crl: edge equivalence}shows that for symmetric configurations, edges that map from one to the other have the same value in the optimal solution, which motivates us to use the distinct edge (stress) classes instead of for the complete graph for the optimization. Before we implement this observation into an efficient algorithm design, we discuss a special case, where the nodes lie on a circle in 2D, forming regular polygons.

\begin{remarkx}[\textit{Stress structure for circular configurations}]
If the nodes lie on a circle and their configuration $\bm{P}$ is ordered sequentially, the permutation matrices $\bm{\Pi}_k$ for $k=1,...,K$ for (\ref{equ: rot perm equivalence}) are limited to cyclic permutations. Then the $\bm{\Omega}_\mathrm{pi}$ is a circulant (also Toeplitz) matrix, as shown in the example in Fig. \ref{fig: generic-usi-edm}.
\end{remarkx}

\subsection{Efficient Algorithm Design}
Motivated by Corollary \ref{crl: edge equivalence}, we would first identify $S$ unique stress classes from the symmetric configuration. However, traversing the matching edge pairs is nontrivial for complicated configurations. It is straightforward to see that having identical edge lengths (Euclidean distances between vertices) is a necessary condition for edge symmetric equivalence. As such, we propose an approximation using the Euclidean distance matrix (EDM) defined by \begin{equation}
        \qty[\edm(\bm{P})]_{ij} = \norm{\bm{p}_i-\bm{p}_j}^2_2 \ \quad i\in\mathcal{V},j\in\mathcal{V},
    \end{equation}
as an initial attempt to identify the $S$ unique stress classes. Edges with the same length shall be considered the same class.

After identifying the unique edge classes, we could create a mapping $\bar{\bm{\omega}} = \bm{S}\underline{\bm{\omega}}$ where $\underline{\bm{\omega}}\in\mathbb{R}^S$ is the reduced stress vector, and $\bm{S}\in\qty{0,1}^{\bar{M}\times S}$ is a selection matrix with $\bm{s}_s \in \qty{0,1}^{\bar{M}}$ denoting the $s$-th column. Then the mapping could be expanded to $\bar{\bm{\omega}} = \sum_{s=1}^S \underline{\omega}_s \bm{s}_s$ where $\underline{\omega}_s$ is the s-th element of $\underline{\bm{\omega}}$. The objective of $\mathcal{P}_2$ (\ref{equ: obj main}) could be substituted into 
\begin{subequations}
    \begin{align}
    \norm{\bar{\bm{\omega}}}_1 - \alpha\bm{\psi}\tp\bar{\bm{\omega}} &= \norm{\bm{S}\underline{\bm{\omega}}}_1 - \alpha\bm{\psi}\tp(\bm{S}\underline{\bm{\omega}})\\
    &= \sum_{s=1}^S(\bm{1}_{\bar{M}}\tp\bm{s}_s)\underline{\omega}_s - \alpha(\bm{S}\tp\bm{\psi})\tp\underline{\bm{\omega}}\\
    &=  \bm{c}\tp\abs{\underline{\bm{\omega}}} - \alpha\underline{\bm{\psi}}\tp\underline{\bm{\omega}},
\end{align}
\end{subequations}
where $\bm{c} = \bm{S}\tp\bm{1}_{\bar{M}}\in\mathbb{Z}_{+}^S$ and $\underline{\bm{\psi}}=\bm{S}\tp\bm{\psi}\in\mathbb{R}^S$ can be precomputed. The PSD constraint (\ref{equ: PSD main}) can be replaced by
\begin{subequations}
\begin{align}
    \bm{\Psi}\diag(\bar{\bm{\omega}})\bm{\Psi}\tp &= \bm{\Psi}\diag\qty(\sum_{s=1}^S \underline{\omega}_s \bm{s}_s)\bm{\Psi}\tp\\
    &= \sum_{s=1}^S\underline{\omega}_s\qty( \bm{\Psi}\diag(\bm{s}_s)\bm{\Psi}\tp)\\
    &= \sum_{s=1}^S\underline{\omega}_s\bm{\Psi}_s,  
\end{align}
\end{subequations}
where $\bm{\Psi}_s = \bm{\Psi}\diag(\bm{s}_s)\bm{\Psi}\tp$ can also be precomputed. Similarly, the spectral norm constraint (\ref{equ: 2norm main}) is
\begin{subequations}
    \begin{align}
        \norm{\bar{\bm{B}}\diag(\bar{\bm{\omega}})\bar{\bm{B}}\tp}_2 &= \sum_{s=1}^S\underline{\omega}_s\qty(\bar{\bm{B}}\diag(\bm{s}_s)\bar{\bm{B}}\tp)\\
        &= \sum_{s=1}^S\underline{\omega}_s\bm{M}_s,
    \end{align}
\end{subequations}
where the basis matrix $\bm{M}_s=\bar{\bm{B}}\diag(\bm{s}_s)\bar{\bm{B}}\tp$ is known. Lastly, the nullspace constraint (\ref{equ: null main}) can be expressed as
\begin{equation}
    \bm{E}\bar{\bm{\omega}} = (\bm{E}\bm{S})\underline{\bm{\omega}} = \underline{\bm{E}}\underline{\bm{\omega}} = \bm{0},
\end{equation}
where $\underline{\bm{E}} = \bm{E}\bm{S}$. Having reduced all the expressions in $\mathcal{P}_2$, we introduce the reduced formulation $\mathcal{P}_3$
\begin{subequations}
  \begin{align}
  & \mathcal{P}_3: \quad \underset{\underline{\bm{\omega}}\in\mathbb{R}^S}{\text{minimize}}
  && \bm{c}\tp\abs{\underline{\bm{\omega}}} - \alpha\underline{\bm{\psi}}\tp\underline{\bm{\omega}} \\
  & \quad \quad \quad \text{subject to}
  &&\sum_{s=1}^S\underline{{\omega}}_s\bm{\Psi}_s - \gamma\bm{I}_{N-D-1} \succeq 0 \\
  &&& \norm{\sum_{s=1}^S\underline{{\omega}}_s\bm{M}_s}_2 \leq \beta\\
  &&& \underline{\bm{E}}\underline{\bm{\omega}} = \bm{0}
  \end{align}
\end{subequations}

A clarifying diagram is provided in (\ref{equ: all formulations}), summarizing all the proposed formulations for the stress design.
\begin{equation}\label{equ: all formulations}
    \mathcal{P}_0 \xRightarrow{\text{relax}} \mathcal{P}_1 \xLeftrightarrow{\text{equivalent}} \mathcal{P}_2 \xLeftrightarrow[\text{configurations}]{\text{symmetric}} \mathcal{P}_3
\end{equation}
Fig. \ref{fig: generic-usi-edm} visualizes a few examples of the stress matrix optimized through $\mathcal{P}_3$ and $\mathcal{P}_2$ in comparison with their corresponding EDMs. As can be seen, the EDM correctly identifies the stress classes in comparison with the generic formulation $\mathcal{P}_2$, and the reduced-size formulation $\mathcal{P}_3$ can find identical solutions. In practice, the EDM approximation often works robustly for symmetric configurations, although it identifies equal or fewer stress classes than the rigorous equivalence requirement.

\begin{figure}[t]
    \centering%
    \includegraphics[width=1\linewidth]{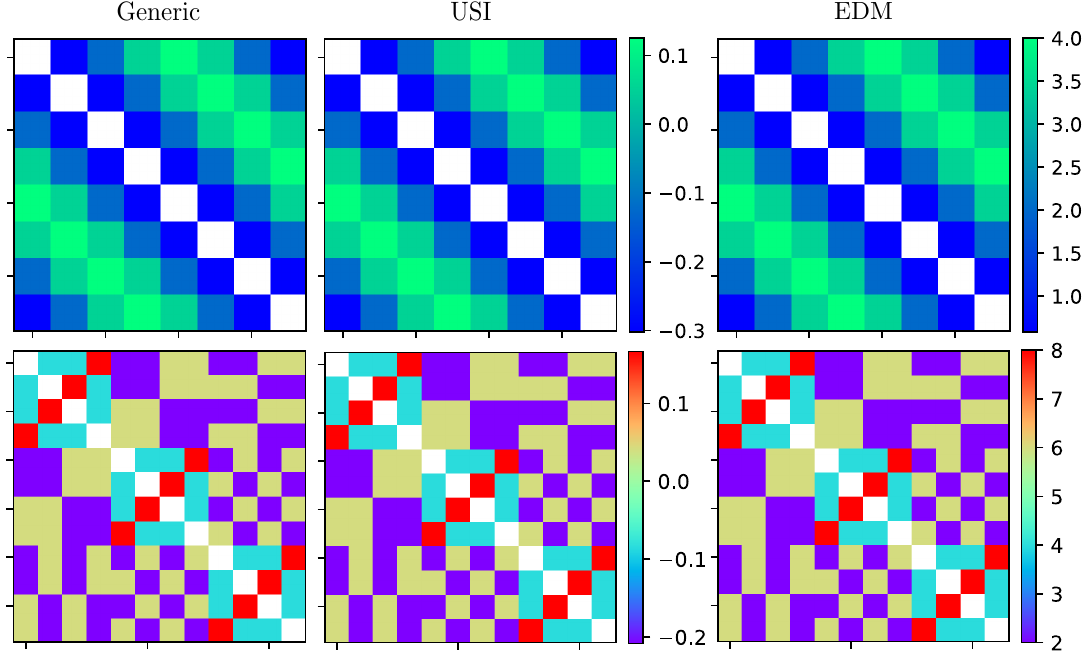}
    \caption{Stress matrices compared with EDMs using an octagon in Fig. \ref{fig: rot-symm configs}(a) (upper) and a cuboctahedron in Fig. \ref{fig: rot-symm configs}(c) (lower), which illustrate the equivalence of the proposed optimization (Generic $\mathcal{P}_2$) and the unique stress identifier (USI) method $\mathcal{P}_3$ on rotationally symmetric configurations. The stress matrices for the octagon are symmetric, Toeplitz, and circulant. The diagonals of matrices are made empty as they do not represent distinct edge classes but are determined by the others. Matching colors represent the same stress values for the first two columns, and the same stress classes to compare with the EDM.  }%
    \label{fig: generic-usi-edm}%
\end{figure}

\subsection{Complexity Analysis}
In general, semidefinite programming can be solved by interior point methods (IPMs) in roughly cubic complexity w.r.t.~the number of decision variables \cite{de2002aspects}, i.e., the size of the stress vector $\bar{\bm{\omega}}$ for $\mathcal{P}_2$, or $\underline{\bm{\omega}}$ for $\mathcal{P}_3$. This is already advantageous to the NP-hard MISDP formulation \cite{xiao2022framework} that scales exponentially. We give a complexity analysis of the proposed $\mathcal{P}_2$ and $\mathcal{P}_3$ w.r.t.~the number of decision variables and the prospects of further reduction using the multicluster framework in the following sections.

Recollect that the generic formulation $\mathcal{P}_2$ has $\bar{M} = \frac{N(N-1)}{2}$ variables for $N$ nodes, which scales quadratically $\mathcal{O}(N^2)$ in terms of the number of variables. For symmetric configurations, this reduces to $\mathcal{O}(S)$ using $\mathcal{P}_3$. For regular 2D polygons, $S$ scales linearly with the number of nodes, hence admitting $\mathcal{O}(N)$ variables. For 3D rotationally symmetric structures, $S$ should be discussed on a case-by-case basis. However, for highly symmetric configurations, such as the Archimedean solids, it is typical that $S\ll N^2$, which permits a massive reduction in the complexity. An example is the truncated icosahedron with $60$ nodes (used in our benchmark as shown in Fig. \ref{fig: benchmark configs} (d)), which has $\bar{M} = 1770$ variables for $\mathcal{P}_2$ but has only $S=21$ variables using our USI for $\mathcal{P}_3$.

Although the rotationally symmetric configurations and their affine transformations cover a large class of nongeneric geometries, $\mathcal{P}_2$ for generic configurations is still inevitable in many cases. The quadratic scaling in SDP variables can be concerning for a large number of nodes. The next section introduces the multicluster control framework, where the large number of nodes is partitioned into smaller clusters, whose stress matrices can be designed independently and more efficiently, which further reduces the overall design complexity.

\section{Multicluster Control of Large Networks}\label{sec: mc control}
In previous sections, we established a convex optimization framework for designing sparse networks with high convergence speed. The generic formulation $\mathcal{P}_2$ scales quadratically with the nodes for an SDP formulation, which becomes demanding for large-scale networks ($N > 100$). In this section, we introduce a divide-and-conquer strategy that partitions a large network into several smaller subnetworks with overlapping nodes, thereby significantly reducing the overall design complexity. Additionally, long-distance communications can also be avoided with the overlapping nodes acting as relays. We propose control strategies with proofs of stability and analyze the associated convergence properties. Furthermore, we derive conditions on the overlapping nodes under which the entire ensemble executes a unified collective affine motion like the single cluster case, as well as conditions under which different clusters exhibit coordinated collective behaviors while preserving individual flexibility. We emphasize that our proposed framework is instantiated using the classic control law (\ref{equ: conventional local law}), but it also applies to its generalizations.

Formally, given a nominal configuration $\bm{P}(\mathcal{V})\in\mathbb{R}^{D\times N}$, we divide it into $C$ subconfigurations $\bm{P}_c(\mathcal{V}_c)\in\mathbb{R}^{D\times N_c}$ for $c=1,...,C$, where $\mathcal{V}_c$ is the subset of nodes for subconfiguration $c$, and $N_c\geq D+1$ is the number of nodes in subconfiguration $c$. We require overlapping nodes among the subconfigurations, formally named \textit{bridging nodes}, such that $\sum_{c=1}^{C}N_c>N$. For each subconfiguration $\bm{P}_c$, a stress matrix $\bm{\Omega}_c\in\mathbb{R}^{N_c\times N_c}$ can be designed using the generic formulation $\mathcal{P}_2$. Then the total number of variables is reduced from $\mathcal{O}(N^2)$ to $\mathcal{O}(\sum_{c=1}^CN_c^2) = \mathcal{O}(CN_c^2)$ if all clusters are assigned an equal number of nodes $N_c$ for a sequential design procedure. If stresses for clusters are designed in parallel, this further reduces to $\mathcal{O}(\underset{c}{\max} N_c^2)$. If $\bm{P}_c$ is symmetric, we could use the accelerated $\mathcal{P}_3$. Let $C_i$ denote the number of clusters node $i$ belongs to, then $C_i=1$ if node $i$ is a nonbridging node, whereas $1<C_i\leq C$ if it is a bridging node.

Based on these clusters, we can now propose a distributed control strategy as detailed in Algorithm~\ref{alg: large scale control}. For agent $i\in\mathcal{V}$ at time $t$, randomly select a cluster $c$ among all possible $C_i$ clusters. Compute the velocity control input for node $i$ related to the selected cluster as 
\begin{equation}\label{equ: local multicluster control}
    \bm{u}_i = -\frac{1}{\pi_{i}^c}\sum_{j\in\mathcal{N}_i^c}\qty[\bm{\Omega}_c]_{ij}\qty(\bm{z}_j-\bm{z}_i),
\end{equation}
where $\mathcal{N}_i^c$ is the neighbor set for agent $i$ in cluster $c$, and the gain $\pi_i^c$ is the probability of agent $i$ choosing cluster $c$. We opt for a discrete uniform distribution, i.e., $\pi_i^c = \frac{1}{C_i}$. For the non-bridging nodes, the selection is deterministic since $C_i=1$.

\begin{algorithm}[t]
    \caption{\textit{Multicluster control}}\label{alg: large scale control}
    \begin{algorithmic}[1]
        \State \textbf{Input}: Configuration  $\bm{P}\in\mathbb{R}^{D\times N}$
        \State \textbf{Initialize}: $\bm{z}_i$ for $i\in\mathcal{V}$
        \State Divide $\bm{P}$ into $\bm{P}_c\in\mathbb{R}^{D\times N_c}$ for $c=1,...,C$
        \State Calculate $\bm{\Omega}_c\in\mathbb{R}^{N_c\times N_c}$ for $c=1,...,C$
        \For {$i\in\mathcal{V}$}
            \State Randomly choose cluster $c$
            \State Exchange for $\bm{z}_j\in\mathcal{N}_j^c$
            \State Calculate $\bm{u}_i = -\frac{1}{\pi_{i}^c}\sum_{j\in\mathcal{N}_i^c}\qty[\bm{\Omega}_c]_{ij}\qty(\bm{z}_j-\bm{z}_i)$
            \State Execute $\dot{\bm{z}}_i = \bm{u}_i$ 
            
        \EndFor

    \end{algorithmic}
\end{algorithm}

\subsection{Stability Analysis}
Controller (\ref{equ: local multicluster control}) represents a linear time-varying (LTV) system, for which the instantaneous system matrix is time-varying and asymmetric. We consider the mean dynamics
\begin{equation}
    \mathbb{E}\qty[\dot{\bm{z}_i}] = \mathbb{E}\qty[\bm{u}_i].
\end{equation}
Taking the expectation for the local dynamics (\ref{equ: local multicluster control}), we obtain
\begin{subequations}
    \begin{align}
    \mathbb{E}[\bm{u}_i] &= \sum_{c=1}^{C} \pi_{i}^c \qty( - \frac{1}{\pi_{i}^c} \sum_{j \in \mathcal{N}_i^c} [\bm{\Omega}_c]_{ij} (\bm{z}_j - \bm{z}_i)) \\
    &= - \sum_{c=1}^{C} \sum_{j \in \mathcal{N}_i^c} [\bm{\Omega}_c]_{ij} (\bm{z}_j - \bm{z}_i).
\end{align}
\end{subequations}
Notice that the probability term $\pi_{i}^c$ cancels with the gain compensation factor. Rewriting this in global form, the mean dynamics become:
\begin{equation}\label{equ: global mean dynamics}
    \mathbb{E}\qty[\dot{\bm{z}}] = - \qty(  (\sum_{c=1}^{C} \tilde{\bm{\Omega}}_c ) \otimes \bm{I}_D ) \mathbb{E}\qty[{\bm{z}}] = - (\bm{\Omega}\otimes \bm{I}_D)\mathbb{E}\qty[{\bm{z}}]
\end{equation}
where $\tilde{\bm{\Omega}}_c \in\mathbb{R}^{N\times N}$ is the expansion of the stress matrix $\bm{\Omega}_c$ into the global dimension with zero padding, and the equivalent global ensemble $\bm{\Omega} = \sum_{c=1}^{C} \tilde{\bm{\Omega}}_c$. A visualization of (\ref{equ: global mean dynamics}) is shown in Fig. \ref{fig: illus mean dynamics}, where the 2 clusters are color coded, and the layers distinguish the stress matrices to be added for the ensemble $\bm{\Omega}$. Note that the mean dynamics yield a linear time-invariant (LTI) system, which provides insight for the convergence property. 

\begin{theoremx}[\textit{Convergence of mean dynamics}]\label{thm: stability of mean dynamics}
    Given $\bm{p}=\vec(\bm{P})$ as a target configuration, the mean dynamics (\ref{equ: global mean dynamics}) exponentially converge to a space containing $\bm{p}$ given arbitrary initialization.
\end{theoremx}
\noindent\begin{proof}
    We first show that the nominal configuration $\bm{p}$ is still an equilibrium for the equivalent ensemble stress matrix $\bm{\Omega}$. The stress matrices $\bm{\Omega}_c$ for cluster $c=1,...,C$ satisfy $(\bm{\Omega}_c\otimes \bm{I}_D)\vec(\bm{P}_c)=\bm{0}$. It also satisfies that $(\tilde{\bm{\Omega}}_c\otimes \bm{I}_D)\bm{p}=\bm{0}$ due to the zero-padding. Then the ensemble $\bm{\Omega}$, the sum of $\tilde{\bm{\Omega}}_c$ over $c=1,...,C$ satisfies $(\bm{\Omega}\otimes \bm{I}_D)\bm{p}=\bm{0}$. 
    
    Now $\tilde{\bm{\Omega}}$ is also a positive-semidefinite (PSD) matrix since it is the sum of $C$ PSD matrices $\tilde{\bm{\Omega}}_c$. As a result, the dynamics (\ref{equ: global mean dynamics}) exhibit exponential convergence. 
\end{proof}

\noindent Theorem \ref{thm: stability of mean dynamics} establishes the stability of the mean dynamics, meaning a non-increasing mean error. We further give insights into the instantaneous dynamics. 
\begin{remarkx}[\textit{Convergence of instantaneous dynamics}]\label{rmk: stability of ins dynamics}
    Let $\sigma(t)$ denote an arbitrary switching pattern representing the agents' random cluster selections at time $t$, and let $\bm{H}_{\sigma(t)}\in\mathbb{R}^{N\times N}$ be the corresponding instantaneous global system matrix, which is a row-wise assembly of $\tilde{\bm{\Omega}}_c$. The continuous-time system dynamics are 
\begin{equation}\label{equ: global ins dynamics}
    \dot{\bm{z}}(t) = -\qty(\bm{H}_{\sigma(t)}\otimes\bm{I}_D)\bm{z}(t).
\end{equation}
Note that $\bm{H}_\sigma(t)$ is not necessarily a PSD matrix to guarantee a nonincreasing error energy using a common Lyapunov function $V(\bm{z}) = \frac{1}{2}\norm{\bm{z}(t)-\bm{z}^*}_2^2$ where $\bm{z}^*\in\mathcal{A}(P)$ is a target configuration in the affine image $\mathcal{A}(P)$. This means that at certain time instants, the agents may temporarily pull each other in conflicting directions, thereby increasing the overall error energy, although this has not been experimentally shown to be an issue.
\end{remarkx}

Note that Theorem \ref{thm: stability of mean dynamics} only establishes the convergence into the space containing the given configuration $\bm{p}$. However, whether this space is only spanned by the affine transformation of $\bm{p}$ needs further investigation, as it determines whether the clusters behave as one affine maneuvering unity.

\subsection{Conditions for Collective Motions}
In this section, we give some insights into the conditions on the bridging nodes for different clusters to adopt collective affine motions, i.e., local affine transformations of each cluster are identical and are thus global. A direct criterion is to evaluate if the rank of the ensemble stress matrix $\bm{\Omega}$ in (\ref{equ: global mean dynamics}) is equal to $N-D-1$. A lower rank implies the introduction of additional flexes beyond affine motions. However, showing the rank of the $\bm{\Omega}$ illustrated in Fig. \ref{fig: illus mean dynamics} is not straightforward. Therefore, we show an alternative view of the problem.

\begin{theoremx}[\textit{Conditions on the bridging nodes for collective motions}]\label{thm: collective motion}
    Suppose node $i\in\mathcal{V}_c$ admits 
    \begin{equation}
        \bm{z}_i = \bm{\Theta}_c\bm{p}_i + \bm{t}_c = \mqty[\bm{\Theta}_c & \bm{t}_c]\mqty[\bm{p}_i\\1],
    \end{equation}
    where $\bm{\bm{\Theta}}_c\in\mathbb{R}^{N_c\times N_c}$ and $\bm{t}_c\in\mathbb{R}^{N_c}$ describe the local affine transformations of cluster $c$. Define the set of bridging nodes $\mathcal{B}_{cc'} = \mathcal{V}_c\cap\mathcal{V}_{c'}$ for any two connected clusters $c,c'=1,...,C$. Then local affine transformations are identical, i.e., $\bm{\Theta}_c = \bm{\Theta}_{c'}$ and $\bm{t}_{c} = \bm{t}_{c'}$, if the augmented configuration of the bridging nodes $\bar{\bm{P}}({\mathcal{B}_{cc'}})\in\mathbb{R}^{(D+1)\times\qty|\mathcal{B}_{cc'}|}$ is full row-rank.
   
\end{theoremx}
\noindent\begin{proof}
    To have identical local transformations requires
\begin{equation}
    \bm{z}_i = \mqty[\bm{\Theta}_c & \bm{t}_c]\mqty[\bm{p}_i\\1] = \mqty[\bm{\Theta}_{c'} & \bm{t}_{c'}]\mqty[\bm{p}_i\\1], \ \forall i\in\mathcal{B}_{cc'}, 
\end{equation}
which yields
\begin{equation}\label{equ: common motion}
    \qty(\mqty[\bm{\Theta}_c & \bm{t}_c] - \mqty[\bm{\Theta}_{c'} & \bm{t}_{c'}])\mqty[\bm{p}_i\\1] = \bm{0}_D,\ \forall i\in\mathcal{B}_{cc'}.
\end{equation}
Stacking for all $i\in\mathcal{B}_{cc'}$ gives $\bm{\Delta}_{cc'}\bar{\bm{P}}({\mathcal{B}_{cc'}}) = \bm{0}_{(D+1)\times\qty|\mathcal{B}_{cc'}|}$, where $\bm{\Delta}_{cc'} = \qty(\mqty[\bm{\Theta}_c & \bm{t}_c] - \mqty[\bm{\Theta}_{c'} & \bm{t}_{c'}])$. Then if $\bar{\bm{P}}({\mathcal{B}_{cc'}})$ is full row-rank, we get the trivial solution $\bm{\Delta}_{cc'}=\bm{0}$.
\end{proof}

\begin{figure}[t]
    \centering%
    \includegraphics[width=0.88\linewidth]{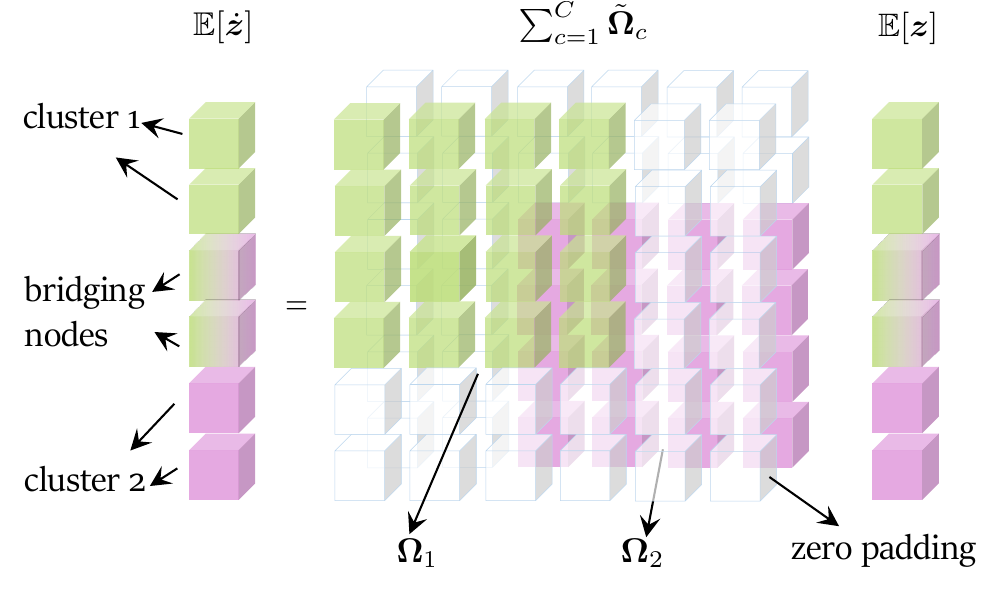}
    \caption{An illustration example of the mean dynamics model (\ref{equ: global mean dynamics}) with $6$ nodes over 2 clusters. The dimension is set to $D=1$ for simplicity.}%
    \label{fig: illus mean dynamics}%
\end{figure}

\noindent Geometrically, affine motions for two clusters $c$ and $c'$ are always aligned if there are at least $D+1$ bridging nodes that fully span $D$-dimensional space for each cluster. If leaders are implemented to guide affine maneuvers, they can be added as required by \cite{zhao2018affine} with no additional restrictions on the distribution to the clusters. 

\subsection{Convergence Speed of the Ensemble Graph}
Since the ensemble graph works as a single cluster under conditions given in Theorem \ref{thm: collective motion}, it is natural to investigate which factors contribute to the convergence speed, i.e., $\lambda_{D+2}(\bm{\Omega})$, given that each cluster is locally optimized for large $\lambda_{D+2}(\bm{\Omega}_c)$.

\begin{theoremx}[\textit{Ensemble convergence speed $\lambda_{D+2}(\bm{\Omega})$}]\label{thm: lambda ensemble}
    If the conditions in Theorem \ref{thm: collective motion} are met for collective motions, the convergence-indicating eigenvalue of the ensemble stress matrix, i.e., $\lambda_{D+2}(\bm{\Omega})$ for $\bm{\Omega}$ in (\ref{equ: global mean dynamics}), is upper bounded by 
    \begin{equation}\label{equ: eigenvalue bound elaborate}
    \lambda_{D+2}(\bm{\Omega}) \le \min_{c \in \{1,...,C\}} \qty( \lambda_{D+2}(\bm{\Omega}_c) + \rho_c).
\end{equation}
Here, defining $\bm{v}_c\in\mathbb{R}^{N_c}$ as the eigenvector associated with $\lambda_{D+2}(\bm{\Omega}_c)$, $\rho_c=\beta\sum_{k\neq c}^C\norm{\bm{v}_{c|\mathcal{B}_{ck}}}_2^2$ with $\beta$ in the spectral norm constraint in $\mathcal{P}_2$, and $\bm{v}_{c|\mathcal{B}_{ck}}$ the subvector of $\bm{v}_c$ corresponding to the bridging nodes to cluster $k$.
\end{theoremx}
\noindent\begin{proof}
    The smallest non-zero eigenvalue can be defined as
\begin{equation}\label{equ: eigenvalue bound}
    \lambda_{D+2}(\bm{\Omega}) = \min_{\substack{\bm{v}\perp\mathrm{Null}(\bm{\Omega}) \\ \norm{\bm{v}}=1}}\bm{v}\tp \bm{\Omega} \bm{v} \leq \frac{\bm{v}\tp \bm{\Omega} \bm{v}}{\bm{v}\tp\bm{v}}
\end{equation}
for any candidate vector $\bm{v}\in\mathbb{R}^{N}$, with $\bm{v}\perp\mathrm{Null}(\bm{\Omega})$. Consider the candidate vector $\tilde{\bm{v}}_c=\mqty[\bm{0}\tp & \bm{v}_c\tp & \bm{0}\tp]\tp\in\mathbb{R}^N$, which is the eigenvector $\bm{v}_c\in\mathbb{R}^{N_c}$ associated with $\lambda_{D+2}(\bm{\Omega}_c)$, correctly zero-padded to length $N$. Lemma \ref{lmm: perp vc nullspace} in Appendix \ref{sec: properties} shows that $\tilde{\bm{v}}_c\perp\mathrm{Null}(\bm{\Omega})$ and $\norm{\tilde{\bm{v}}_c}=1$, and thus $\tilde{\bm{v}}_c$ satisfies the bound (\ref{equ: eigenvalue bound}), i.e., $\lambda_{D+2}(\bm{\Omega})\leq \tilde{\bm{v}}_c\tp\bm{\Omega}\tilde{\bm{v}}_c$. 
Further extending this bound, we obtain
\begin{subequations}\label{equ: lmd inequality}
\begin{align}
     \lambda_{D+2}(\bm{\Omega})\leq\tilde{\bm{v}}_c\tp \bm{\Omega} \tilde{\bm{v}}_c &= \tilde{\bm{v}}_c\tp \qty(\tilde{\bm{\Omega}}_c + \sum_{k\neq c}^{C}\tilde{\bm{\Omega}}_k) \tilde{\bm{v}}_c \\ &=  \tilde{\bm{v}}_c\tp \tilde{\bm{\Omega}}_c \tilde{\bm{v}}_c +  \tilde{\bm{v}}_c\tp  (\sum_{k\neq c}^C\tilde{\bm{\Omega}}_k) \tilde{\bm{v}}_c,
\end{align}
\end{subequations}
with
\begin{subequations}
    \begin{align}
        \tilde{\bm{v}}_c\tp \tilde{\bm{\Omega}}_c \tilde{\bm{v}}_c &= \bm{v}_c\tp \bm{\Omega}_c \bm{v}_c = \lambda_{D+2}(\bm{\Omega}_c) \\
        \tilde{\bm{v}}_c\tp (\sum_{k\neq c}^C\tilde{\bm{\Omega}}_k)\tilde{\bm{v}}_c &= \sum_{k\neq c}^C\bm{v}_{c|\mathcal{B}_{ck}}\tp\bm{\Omega}_{k|{\mathcal{B}_{ck}}}\bm{v}_{c|\mathcal{B}_{ck}}\label{equ: ensemble lambda part 2},
    \end{align}
\end{subequations}
where $\bm{v}_{c|\mathcal{B}_{ck}}$ denotes the sub-vector of $\bm{v}_c$ corresponding to the bridging nodes $\mathcal{B}_{ck}$ shared between cluster $c$ and $k$, and $\bm{\Omega}_{k|{\mathcal{B}_{ck}}}$ is the subblock of $\bm{\Omega}_k$ corresponding to those bridging nodes. We could further relax (\ref{equ: ensemble lambda part 2}) using
\begin{subequations}
\begin{align}
\bm{v}_{c|\mathcal{B}_{ck}}\tp\bm{\Omega}_{k|{\mathcal{B}_{ck}}}\bm{v}_{c|\mathcal{B}_{ck}} &\leq \lambda_{\mathrm{max}}\qty(\bm{\Omega}_{k|{\mathcal{B}_{ck}}})\norm{\bm{v}_{c|\mathcal{B}_{ck}}}_2^2\\
&\leq\lambda_{\mathrm{max}}\qty(\bm{\Omega}_k)\norm{\bm{v}_{c|\mathcal{B}_{ck}}}_2^2\label{equ: cauchy interlacing relax}\\
& \leq \beta \norm{\bm{v}_{c|\mathcal{B}_{ck}}}_2^2,\label{equ: beta relax}
\end{align}
\end{subequations}
where (\ref{equ: cauchy interlacing relax}) is the relaxation by the Cauchy interlacing theorem \cite{hwang2004cauchy} which upper bounds the maximum eigenvalue of a principal submatrix by that of the parent matrix, and (\ref{equ: beta relax}) is spectral norm constraint of $\mathcal{P}_2$ in the stress design. Substituting (\ref{equ: beta relax}) back to (\ref{equ: ensemble lambda part 2}) gives us $\rho_c=\beta\sum_{k\neq c}^C\norm{\bm{v}_{c|\mathcal{B}_{ck}}}_2^2$. Finally, since this inequality (\ref{equ: lmd inequality}) is satisfied for all $c=1,...,C$, we adopt the tightest one, i.e., bound (\ref{equ: eigenvalue bound elaborate}).
\end{proof}

With the upper bound (\ref{equ: eigenvalue bound elaborate}) we can infer that the convergence speed of the ensemble graph is bottlenecked by two factors: the slowest cluster indicated by $\lambda_{D+2}(\bm{\Omega}_c)$, and the coupling strength of the bridging nodes indicated by $\rho_c$. This provides some guidelines for the design of the clusters, i.e., to make sure each cluster converges fast using the proposed algorithm $\mathcal{P}_2$, and to ensure the bridging nodes connect the clusters well. The design of the optimal bridging remains an open problem; however, heuristics such as using more bridging nodes for a fixed $\beta$ are proven to be effective in practice.

\begin{figure}[t]
    \centering
    \subfloat[\scriptsize Bridging nodes span 2D]{\raisebox{6mm}{ 
        \includegraphics[width=0.154\textwidth]{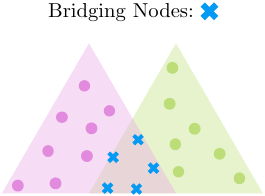}
    }}
    \hspace{0pt} 
    \subfloat[\scriptsize Rank-deficient bridge]{\raisebox{0mm}{
        \includegraphics[width=0.126\textwidth]{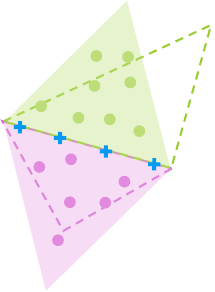}
    }}
    \hspace{0pt} 
    \subfloat[\scriptsize Rank-deficient bridge]{\raisebox{0mm}{
        \includegraphics[width=0.154\textwidth]{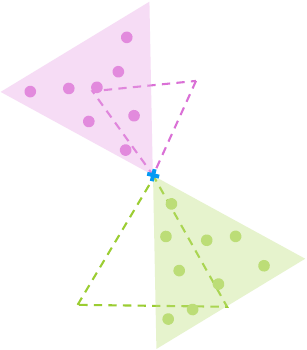}
    }}
    
    \caption{An illustration of local flexibility introduced by rank-deficient bridging nodes. The shaded area outlines the desired configuration and the dashed geometry shows the potential ambiguity.}
    \label{fig: partial-collective motions}
\end{figure}

\subsection{Extension to Partial-collective Behaviors}
Following the conclusion of the previous subsection, if the bridging nodes are not selected such that they span $D$-dimensional space, the clusters may not be aligned in motion. Instead, the swarm can be intentionally designed to exhibit partial-collective motions, allowing for cluster-wise flexibility. In the most extreme case, if no bridging nodes are present, the clusters act as completely independent systems. 

\begin{table*}[t]
\scriptsize
\centering
    \begin{threeparttable}
	\caption{Benchmark Results on Runtime (in [Seconds]) }
	\label{tab: benchmark results}
	\centering
	\begin{tabular*}{\linewidth}{@{\extracolsep{\fill}} c c c c c c c c c c}
		\toprule
       test case & Random-5 & Random-8 & Random-50 & Random-100 & Random-500 & Circular-10 & Cuboctahedron & Trunc. Icosa. & Letter-``W" \\
       $N$ & 5 & 8 & 50 & 100 & 500 & 10 & 12 & 60  & 260\\
       $D$ & 2 & 2 & 2 & 2 & 2 & 2 & 3 & 3 & 3\\
		\midrule
		proposed ($\mathcal{P}_2$) & 6.7e-3 & 10.3e-3 & 4.8 & N/A & N/A & 13.8e-3 & 16.3e-3  & 7.7 & N/A \\
        proposed-usi ($\mathcal{P}_3$) & N/A & N/A & N/A & N/A & N/A & 8.9e-3 & 8.6e-3 & 0.27  & N/A \\
        proposed-m.c.\tnote{1} & -- & -- & -- & 17.6 & 88.5 & -- & -- & --  & 18.5\tnote{2} \\
        \\
        Lin et al. & 4.1e-3 & 7.1e-3 & N/A &  N/A & N/A & 9.5e-3 & N/A & N/A  & N/A\\
        Yang et al. & 0.080e-3 & 0.10e-3 & 0.55e-3 & 0.97e-3 & 4.8e-3 & 0.2e-3 & 0.1e-3 & 0.6e-3  & 2.5e-3\\
        Xiao et al. & 519.0e-3 & 108.3 & N/A & N/A & N/A & 322.7  & N/A & N/A & N/A \\
        Li et al. & 0.307e-3 & 0.56e-3 & 2.7e-3 & 5.3e-3 & 31.1e-3 & 0.65e-3 & 0.7e-3 & 3.7e-3 & 15.2e-3\\
        
		\bottomrule
	\end{tabular*}
    \begin{tablenotes}
            \scriptsize
            \item[] ``N/A": Not applicable due to memory constraints, impracticably long runtime, or invalid results. ``--": Multicluster operations not needed.
            \item[1] The proposed-m.c. is configured with 2 clusters of 60 nodes for ``Random-100", 10 clusters of 60 nodes for ``Random-500", and 4 clusters of 58 nodes for ``Letter-W".
            \item[2] The computation is only needed for 1 cluster, as all clusters are identical up to affine transformations.
        \end{tablenotes}
    \end{threeparttable}
\end{table*}

To maneuver an affine formation, the global ensemble fundamentally requires a base set of at least $D+1$ affinely independent leaders to span the space \cite{zhao2018affine}. However, when the swarm features rank-deficient bridges, this global leader set alone is insufficient to uniquely determine the configuration of every sub-cluster. The bridging nodes $\mathcal{V}_{\mathcal{B}_c}$ within cluster $c$ act as virtual leaders, providing constraints dictated by the neighboring cluster. To guarantee that cluster $c$ is affinely localizable, meaning its local shape and affine transformations are fully determined and controllable, the actual leaders within the cluster must compensate for the bridge's rank deficiency. They must jointly satisfy the local condition
\begin{equation}\label{equ: leader condition}
\rank\qty(\bar{\bm{P}}\qty({\qty(\mathcal{V}_l\cap\mathcal{V}_c)\cup\mathcal{V}_{\mathcal{B}_c}})) = D+1,\ \ \forall c=1,...,C
\end{equation}
where $\mathcal{V}_l$ is the global subset of chosen leaders, and $\mathcal{V}_{\mathcal{B}_c}$ is the set of bridging nodes attached to cluster $c$.
Physically, (\ref{equ: leader condition}) dictates how to control the unconstrained degrees of freedom with leaders introduced by a rank-deficient bridge. Several 2D examples are visualized in Fig. \ref{fig: partial-collective motions}, where two clusters with different bridging node setups are considered. In Fig. \ref{fig: partial-collective motions} (a), the bridging nodes satisfy conditions for collective behaviors, thus no additional leaders are required beyond the $3$ for global maneuver. In (b), the bridging nodes form a line, introducing loose local degrees of freedom, as can be seen by the dashed outline. In this case, at least $1$ leader should be allocated to each cluster. Similarly, (c) has 1 bridging node and requires $2$ leaders for each cluster to ensure global localizability.

\begin{figure}[t]
	\centering	

        \subfloat[\scriptsize Circular-10 in 2D]{\raisebox{0ex}
		{\includegraphics[width=0.15\textwidth]{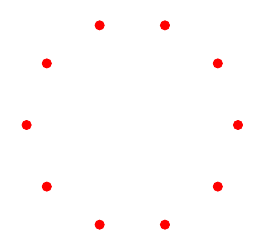}}%
	}
        \hspace{9mm}
	\subfloat[\scriptsize Cuboctahedron in 3D]{\raisebox{0ex}
		{\includegraphics[width=0.15\textwidth]{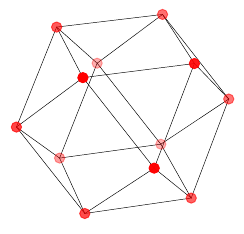}}%
	}\\
    \hspace{0mm}
        \subfloat[\scriptsize 
        Truncated-Icosahedron in 3D]{\raisebox{0ex}
    		{\includegraphics[width=0.19\textwidth]{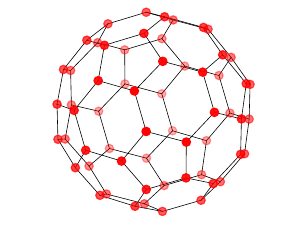}}%
    	}
        \hspace{3mm}
        \subfloat[\scriptsize Letter-``W" in 3D]{\raisebox{0ex}
    		{\includegraphics[width=0.19\textwidth]{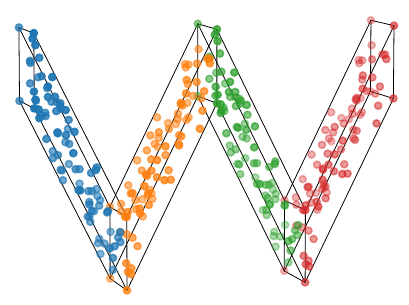}}%
    	}

	\caption{The special configurations used for the benchmark. For 3D configurations, the edges only outline the geometry for clear visualization.}
	\label{fig: benchmark configs}
		
\end{figure}



\section{Simulations}\label{sec: simulations}
In this section, we numerically validate the proposed algorithms and the multicluster control framework. We first perform a benchmark comparing with multiple state-of-the-art stress matrix design algorithms on several performance metrics. We then demonstrate the convergence and the inter-cluster flexibility of our proposed multicluster control. All code and plots are publicly available\footnote{\url{https://github.com/asil-lab/zli-large-afc}}.

\subsection{Benchmarks of Stress Design}
The algorithms are benchmarked on a Linux platform with an AMD Ryzen 9 7900X processor and 64GB of DRAM. We select some representative state-of-the-art candidates from all existing methods, including Lin et al. \cite{lin2015necessary}, Xiao et al. \cite{xiao2022framework}, Yang et al. \cite{yang2018constructing}, and Li et al. \cite{li2025distributed}. The proposed solutions include the generic stress design method $\mathcal{P}_2$ from Section \ref{sec: methodology} with 3 sets of hyperparameters (namely $\gamma=0.1$, $\beta=1$, and $\alpha\in\qty{0.5, 1.5, 5}$), the unique stress identifier (USI) method $\mathcal{P}_3$ from Section \ref{sec: usi} for reduced complexity, and the multicluster division from Section \ref{sec: mc control}. The optimization-based algorithms are solved using MOSEK \cite{aps2020mosek} with the CVXPY interface \cite{diamond2016cvxpy}, and all algorithms perform 100 Monte Carlo runs for each case. We test our solutions using random configurations of different sizes, namely $N \in \qty{5, 8, 50, 100, 500}$ with a naming convention such as ``Random-5", ``Random-8", etc. Some rotationally symmetric configurations are considered as well as shown in Fig. \ref{fig: benchmark configs}, including ``Circular-10", which is a $10$-node regular decagon, a $12$-node ``Cuboctahedron", and a $60$-node ``Truncated Icosahedron". Lastly, we use a customized multicluster configuration ``Letter-W" in Fig. \ref{fig: benchmark configs}(d) to demonstrate the partial-collective behaviors of the multicluster control, where the 4 segments of the configurations are identical up to affine transformations.

Our performance metrics include (a) the average runtime of the algorithms, which indicates their complexity and scalability; (b) the average degree $2M/N$ of the resulting graph, which reflects the sparsity and the general communication overhead; (c) eigenvalue $\lambda_{D+2}(\bm{\Omega})$, which indicates the convergence speed of the formation; and (d) the spectral efficiency 
\begin{equation}\label{equ: def spec eff}
    \eta = \frac{\lambda_{D+2}(\bm{\Omega}) N^2}{\lambda_{N}(\bm{\Omega})M},
\end{equation}
which assesses the convergence speed, given the same communication capacity. For the multicluster setup, we evaluate the equivalent stress matrix $\bm{\Omega} = \sum_{c=1}^{C} \tilde{\bm{\Omega}}_c$ in the mean dynamics. The results of the runtime are reported in Table \ref{tab: benchmark results} and statistics of other metrics are shown in Fig. \ref{fig: eigs and spec eff}.

We first evaluate the complexity of the algorithms through the average runtime in Table \ref{tab: benchmark results}, from which we generally observe that the optimization-based algorithms consume a longer time to compute, compared with the analytical Yang et al. and Li et al.  In particular, Lin et al. is constrained by the (manual) prescription of the graph, and Xiao et al. (MISDP) is only capable of small-scale networks due to the NP-hard complexity. The generic version of the proposed algorithm is capable of handling medium-sized networks ($N<100$), above which the multicluster division results in a reasonable runtime. In the case of rotationally symmetric configurations, such as circular and cuboctahedron, the USI method with $\mathcal{P}_3$ effectively reduces the runtime compared to the generic solution $\mathcal{P}_2$. In case of the truncated icosahedron, the USI method gives a roughly $30\times$ speedup. For special geometries such as Fig. \ref{fig: benchmark configs}(d), where a subconfiguration is repeated multiple times up to affine transformations, it is natural to divide them into several clusters, for which the stress matrices also repeat and only need to be computed once.

The sparsity of the graph is shown in the top figure in Fig. \ref{fig: eigs and spec eff}. Existing state-of-the-art methods typically yield sparse solutions, whereas the proposed approach allows explicit control over sparsity through the hyperparameter $\alpha$. Particularly, for small values (i.e., $\alpha = 0.5$, which is around the critical value $\alpha^*$ introduced in Section \ref{sec: hparams}), it produces graphs that are comparably sparse to those obtained by state-of-the-art methods.  For a large $\alpha = 5$, it approaches the upper bound corresponding to a complete graph. Note that in the multicluster implementations of ``Random-100" and ``Random-500", this upper bound cannot be reached even for large $\alpha$, since non-bridging nodes among the clusters do not interact. This feature naturally promotes a sparse solution and can further be adapted to accommodate, e.g., limits of communication ranges.

The convergence speed and spectral efficiency are also evaluated in Fig. \ref{fig: eigs and spec eff}. From the smallest non-zero eigenvalue $\lambda_{D+2}(\bm{\Omega})$, we observe that the analytical algorithms of Yang et al. and Li et al. yield values that are 1 or 2 orders of magnitude smaller for medium- to large-scale networks. This may be considered a trade-off of their sparse nature. The optimization-based algorithms, including the proposed, generally yield larger $\lambda_{D+2}(\bm{\Omega})$, meaning a faster convergence speed since systematic searches for optimality are involved. Note that for the proposed algorithm with $\alpha=0.5$, the convergence is much faster than the analytical solutions while achieving similar levels of sparsity across the various cases. This is also seen from the spectral efficiency plot in the bottom figure of Fig. \ref{fig: eigs and spec eff}, where the proposed algorithms better exploit the communication capacity into accelerating the convergence more efficiently. With $\alpha$ tuned larger, the resulting graph becomes denser while further accelerating the convergence. However, larger $\alpha$ may be less favored in practice as the convergence boost is increasingly marginal, while compromising the graph sparsity.

To summarize the evaluation on stress design, our proposed algorithms have sparse results, fast convergence, and high spectral efficiency while admitting acceptable complexity for medium and large-scale networks. Additionally, they are robust to generic and nongeneric configurations in 2D and 3D, and yield consistently superior results as compared to state of the art across various benchmarks. 

\begin{figure}[t]
    \centering%
    \includegraphics[width=1.0\linewidth]{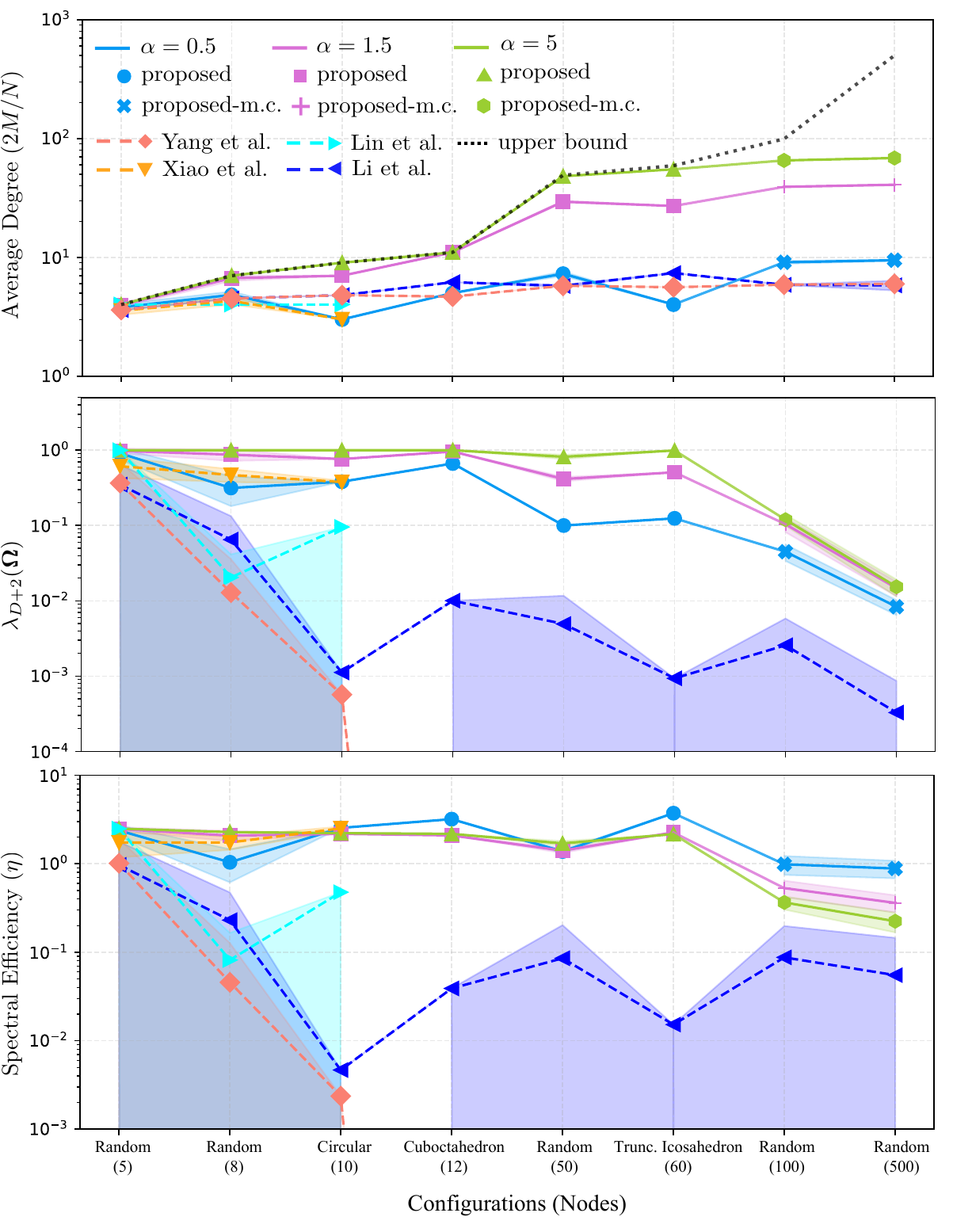}
    \caption{Benchmark performances of the proposed algorithms and the state-of-the-arts. The plotted metrics include the graph sparsity characterized by average degree of the graph $2M/N$ (upper), the convergence-speed-indicating eigenvalue $\lambda_{D+2}(\bm{\Omega})$ (middle), and the spectral efficiency (\ref{equ: def spec eff}) (lower).}%
    \label{fig: eigs and spec eff}%
\end{figure}

\subsection{Multicluster Control with Collective Motions}

We now validate the proposed multicluster control framework using the ``Random-100" test case. Recall from Section \ref{sec: mc control} that a certain number of bridging nodes are required for the clusters to engage in collective motions, which also affects the convergence speed in the mean dynamics. Fig. \ref{fig: mc control results}(a) demonstrates an example of 2 clusters with different numbers of bridging nodes, and (b) shows the convergence speed through the eigenvalue $\lambda{_{D+2}(\bm{\Omega})}$ of the cluster ensemble in the left plot, and the tracking error convergence on the right.

Since the ``Random-100" is a 2D configuration, the minimum number of bridging nodes needed is $D+1=3$ and they must span the 2D space. It can be clearly seen from Fig. \ref{fig: mc control results}(b) that $\lambda_{D+2}(\bm{\Omega})$ is $0$ and the control tracking error does not decrease when there are only $2$ bridging nodes. This shows the lack of collective behaviors due to the loose degrees of freedom caused by insufficient bridging nodes. As we increase the bridging nodes, we observe larger $\lambda_{D+2}(\bm{\Omega})$, which suggests a faster convergence. Correspondingly, the tracking error presents steeper convergence curves with more bridging nodes. 

\begin{figure}[t]
    \centering
    \subfloat[\scriptsize "Random-100" configuration with 2 clusters. The number of bridging nodes is varied among 2 (not sufficient for collective motion), 5, 10, and 20.]{\raisebox{0mm}{ 
        \includegraphics[width=0.49\textwidth]{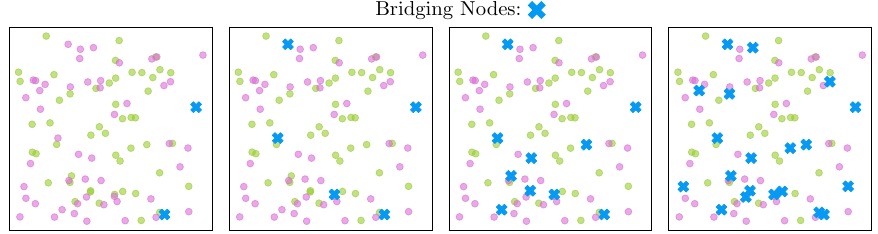}
    }}\\
    \vspace{2pt}
    \hspace{0pt} 
    \subfloat[\scriptsize The convergence speed through the eigenvalues (left, only plotted the first 10) and the control tracking error (right).]{
        \includegraphics[width=0.485\textwidth]{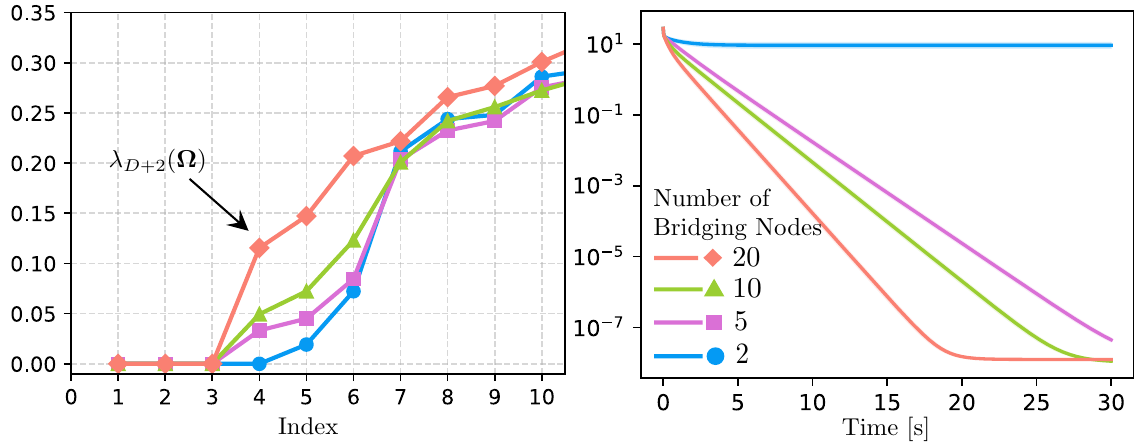}
    }
    
    \caption{The results of the multicluster control with collective affine motions. The number of bridging nodes is varied to validate the claim about the upper bound on the ensemble graph's eigenvalues.}
    \label{fig: mc control results}
\end{figure}

\subsection{Swarm with Inter-Cluster Flexibility}
We finally demonstrate the feature of inter-cluster flexibility of the multicluster control framework using the ``Letter-W" configuration as shown in Fig. \ref{fig: benchmark configs}(d). The configuration is constructed so that each letter segment (considered a cluster) is identical up to an affine transformation, allowing the stress matrix to be computed for only a single segment. To allow intercluster flexibility, the clusters are interfaced with 4 nodes on a plane, and the remaining degrees of freedom are tightened by leaders within each cluster. Then by applying local affine transformations, the clusters may present other geometries than the ``Letter-W" by simply maneuvering the leader positions. Fig. \ref{fig: demo traj} demonstrates the maneuvering patterns, where all nodes are initialized randomly in 3D, and they sequentially converge to the target formation of a ``Bar"" shape, followed by a ``Letter-V" shape, and finally the ``Letter-W" while admitting a collective translation movement. This validates the feature of inter-cluster flexibility with partial collective motions of the proposed multicluster control framework.

\section{Conclusion}\label{sec: conclusions}
In this work, we propose a convex formulation for designing the stress matrix used for affine formation control, aiming to optimize the convergence speed while keeping the network sparse. For the rotationally symmetric configurations, we show that the optimal solution is highly structured, in which the stress values are identical for equivalent edges, and thus propose a reduced formulation for lower computation complexity. Compared with the state-of-the-art solutions through a comprehensive benchmark, we conclude that our proposed formulation is superior in the convergence speed and sparsity while maintaining acceptable computation complexity. Having the optimal design of the network, we also propose a multicluster control framework in which large-scale networks are divided into smaller networks with the overlapping nodes configured to present both collective and noncollective behaviors.  
In future work, to realize large swarm implementations, we aim to consider large-scale collision avoidance, explicit modeling and optimization of the communication cost. As a direct application of the multicluster framework, extensions of the stress-based control used in this work can be used for relevant scenarios. Furthermore, the selection of the bridging nodes can also be optimized for collective motions.

\begin{figure}[t]
    \centering%
    \includegraphics[width=0.95\linewidth]{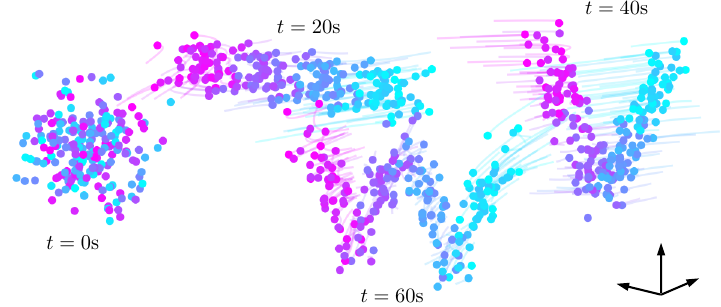}
    \caption{Large-scale formation maneuvers with inter-cluster flexibility. The swarm is initialized at a random configuration in 3D. Multicluster control is deployed to achieve the ``Letter-W" and several other configurations using the same optimized stress matrix.}%
    \label{fig: demo traj}%
\end{figure}

\appendices
\section{Proofs of Properties and Lemmas}\label{sec: properties}
\begin{propertyx}[\textit{Permutation invariance of $\tr(\bm{\Omega})$}]\label{prpt: trace invariance} The trace of the stress matrix is invariant under permutations, i.e., $\tr\qty(\bm{\Omega}) = \tr\qty(\bm{\Pi}_k\bm{\Omega}\bm{\Pi}_k\tp)$  
\end{propertyx}

\noindent\begin{proof}
Using the cyclic property of the trace operator, $\tr\qty(\bm{\Pi}_k\bm{\Omega}\bm{\Pi}_k\tp) = \tr\qty(\bm{\Pi}_k\tp\bm{\Pi}_k\bm{\Omega}) = \tr(\bm{\Omega})$ since permutation matrices are orthogonal, i.e., $\bm{\Pi}_k\tp\bm{\Pi}_k = \bm{I}_N$.
\end{proof}
\\

\noindent \textbf{Proof of Lemma \ref{lmm: prob perm invariance}}.

\noindent\begin{proof}
    We first show the invariance property of the objective function. It is straightforward to see that $\norm{\bar{\bm{\Pi}}_{\mathcal{E},k}\bar{\bm{\omega}}}_1 = \norm{\bar{\bm{\omega}}}_1$ for any permutation $\bar{\bm{\Pi}}_{\mathcal{E},k}\in\qty{0,1}^{\bar{M}\times\bar{M}}$ since permuting the elements of $\bar{\bm{\omega}}$ does not change the L1 norm. The trace of the stress matrix $\tr(\bm{\Omega})$ is also permutation invariant, shown in Property \ref{prpt: trace invariance} in Appendix \ref{sec: properties}. We could thus conclude that the objective (\ref{equ: obj func}) satisfies permutation invariance in Definition \ref{def: obj invariance}.

    Next, we show the invariance property of the feasible set. Observe from (\ref{equ: node edge perm}) that the permutations of the stress matrix $\bm{\Pi}_k\bm{\Omega}{\bm{\Pi}_k}\tp$ is a similarity transforms of $\bm{\Omega}$, which preserves the eigenvalues. As such, Constraints (\ref{equ: cstr max eig}) and (\ref{equ: cstr min eig}) that limit the range of the eigenvalues still hold under permutations. As for the equality constraint (\ref{equ: null cstr mat}), we extend the definition (\ref{equ: rot perm equivalence}) to $\bar{\bm{R}}_k\bar{\bm{P}} = \bar{\bm{P}}\bm{\Pi}_k$ where $\bar{\bm{R}}_k = \diag\qty(\bm{R}_k,1)\in\mathbb{R}^{(D+1)\times (D+1)}$. Observe that
    \begin{subequations}
    \begin{align}               \bar{\bm{B}}\diag\qty(\bar{\bm{\Pi}}_{\mathcal{E},k}\bar{\bm{\omega}})\bar{\bm{B}}\tp\bar{\bm{P}}\tp &\stackrel{\text{(\ref{equ: node edge perm})}}{=} \bm{\Pi}_k\bm{\Omega}{\bm{\Pi}_k}\tp\bar{\bm{P}}\tp\\
    &\stackrel{}{=} \bm{\Pi}_k\bm{\Omega}\bar{\bm{P}}\tp\bar{\bm{R}}_k\tp\\
    &\stackrel{\text{(\ref{equ: null cstr mat})}}{=}\bm{0},
    \end{align}
    \end{subequations}
which proves the permutation invariance of constraint (\ref{equ: null cstr mat}). As such, problem $\mathcal{P}_1$ is permutation-invariant.
\end{proof}
\\



\begin{lemmax}[\textit{Orthogonality of $\tilde{\bm{v}}_c$ and $\mathrm{Null}(\bm{\Omega})$}]\label{lmm: perp vc nullspace}
    Consider $\tilde{\bm{v}}_c=\mqty[\bm{0}\tp & \bm{v}_c\tp & \bm{0}\tp]\tp\in\mathbb{R}^N$,  which contains the eigenvector $\bm{v}_c\in\mathbb{R}^{N_c}$ associated with $\lambda_{D+2}(\bm{\Omega}_c)$, zero-padded to the length of $N$. Then $\tilde{\bm{v}}_c\perp\mathrm{Null}(\bm{\Omega})$ if $\rank(\bm{\Omega}) = N-D-1$.
\end{lemmax}
\noindent\begin{proof}
    If $\rank(\bm{\Omega}) = N-D-1$, then augmented configuration matrix $\bar{\bm{P}}(\mathcal{V})$ fully span $\mathrm{Null}(\bm{\Omega})$. Then we could verify that $\tilde{\bm{v}}_c\tp\bm{1}_N = \bm{v}_c\tp\bm{1}_{N_c} = 0$ since vector $\bm{1}_{N_c}$ is in the nullspace of $\bm{\Omega}_c$ and is thus orthogonal to $\bm{v}_c$. Similarly, we could verify that $\tilde{\bm{v}}_c\tp\bm{P}\tp = \bm{v}_c\tp\bm{P}_c\tp=\bm{0}_D\tp$. Therefore, we could establish that $\tilde{\bm{v}}_c\perp\mathrm{Null}(\bm{\Omega})$ and $\tilde{\bm{v}}_c\tp\tilde{\bm{v}}_c=1$ since $\tilde{\bm{v}}_c$ is zero-padded.
\end{proof}

\bibliographystyle{IEEEtran}
\bibliography{refs}

\end{document}